\documentclass[11pt,a4wide]{report}
\usepackage{times}
\usepackage{graphicx}
\usepackage{latexsym}

\usepackage{color}
\usepackage{graphicx}

\usepackage{xspace}
\newcommand{\eg}{e.g.,\xspace}

\usepackage{authblk}

\usepackage{amsmath}
\usepackage{amsfonts}
\usepackage{amsthm}
\newtheorem{definition}{Definition}
\newtheorem{proposition}{Proposition}
\newtheorem{theorem}{Theorem}
\newtheorem{lemma}{Lemma}
\newtheorem{corollary}{Corollary}

\usepackage{algorithm}
\usepackage[noend]{algpseudocode}

\algnewcommand\algorithmicswitch{\textbf{switch}}
\algnewcommand\algorithmiccase{\textbf{case}}
\algdef{SE}[SWITCH]{Switch}{EndSwitch}[1]{\algorithmicswitch\ #1\ \algorithmicdo}{\algorithmicend\ \algorithmicswitch}%
\algdef{SE}[CASE]{Case}{EndCase}[1]{\algorithmiccase\ #1}{\algorithmicend\ \algorithmiccase}%
\algtext*{EndSwitch}%
\algtext*{EndCase}%

\newcommand{\atl}{ATL\xspace}
\newcommand{\atlr}{RB$\pm$ATL\xspace}

\newcommand\RBATL{RB-ATL\xspace}
\newcommand\nat{\mathbb{N}}
\newcommand\natinfty{\nat_{\infty}}
\newcommand\integer{\mathbb{Z}}
\newcommand{\llangle}{\langle\!\langle}
\newcommand{\rrangle}{\rangle\!\rangle}
\newcommand{\ceil}[1]{\lceil #1 \rceil}
\newcommand\AO[2]{\llangle #1 \rrangle \!\bigcirc\! #2}
\newcommand\AG[2]{\llangle #1 \rrangle \Box #2}
\newcommand\AU[3]{\llangle #1 \rrangle #2 \,{\cal U}\, #3}

\newcommand\idle{idle}

\newcommand\ralAO[3]{\llangle #1 \rrangle^{#2} \!\bigcirc\! #3}
\newcommand\ralAG[3]{\llangle #1 \rrangle^{#2} \Box #3}
\newcommand\ralAU[4]{\llangle #1 \rrangle^{#2} #3 \,{\cal U}\, #4}

\newcommand{\Agt}{Agt}
\newcommand{\Res}{Res}
\newcommand{\Act}{Act}

\newcommand\PNtran[1]{\,[{#1}\rangle\,}

\newcommand{\zerob}{{\bar 0}}
\newcommand{\linfarrow}{\mathrel{\smash{\stackrel{\infty}{\leftarrow}}}}
\newcommand{\projinf}[2]{{#1 \linfarrow #2}}

\newcommand{\atlop}[1]{\llangle #1 \rrangle}

\newcommand{\atlbox}[1]{\atlop{#1}\square}
\newcommand{\atlu}[3]{\atlop{#1} #2\;\!\mathcal{U} #3 }

\newcommand{\ral}{RAL\xspace}
\newcommand{\pfral}{pr-rf-\ral'\xspace}
\newcommand{\pfralrelax}{pr-rf-\ral''\xspace}
\newcommand{\atlrnt}{\atlr-nt\xspace}

\newcommand{\wall}{\begin{aligned}[t] &}

\newcommand{\return}{\end{aligned}}
\newcommand{\En}{\textsf{\small En}}
\newcommand{\Share}{\textsf{\small Share}}
\newcommand{\sh}{\textsf{\small sh}}
\newcommand{\produce}{\textsf{\small prod}}
\newcommand{\cons}{\textsf{\small cons}}
\newcommand{\tr}{\textit{tr}}

\begin{document}

\title{Technical Report: \\
Model-Checking for Resource-Bounded ATL with Production and Consumption of Resources}

\author[*]{Natasha Alechina}
\author[*]{Brian Logan}
\author[*]{\\Hoang Nga Nguyen}
\author[**]{Franco Raimondi}

\affil[*]{School of Computer Science, The University of Nottingham, UK\\
{\texttt{\{nza,bsl,hnn\}@cs.nott.ac.uk}}}
\affil[**]{Department of Computer Science, Middlesex University, UK\\
{\texttt{f.raimondi@mdx.ac.uk}}}

\maketitle

\begin{abstract}
Several logics for expressing coalitional ability under resource bounds have been proposed and studied in the literature. Previous work has shown that if only consumption of resources is considered or the total amount of resources produced or consumed on any path in the system is bounded, then the model-checking problem for several standard logics, such as Resource-Bounded Coalition Logic (RB-CL) and Resource-Bounded Alternating-Time Temporal Logic (RB-ATL) is decidable. However, for coalition logics with unbounded resource production and consumption, only some undecidability results are known. In this paper, we show that the model-checking problem for RB-ATL with unbounded production and consumption of resources is decidable but EXPSPACE-hard. We also investigate some tractable cases and  provide a detailed comparison to a variant of the resource logic RAL, together with new complexity results.
\end{abstract}

\tableofcontents

\section{Introduction}
\label{sec:intro}
Alternating-Time Temporal Logic (ATL) \cite{Alur//:02a} is widely used
in verification of multi-agent systems. ATL can express properties
related to coalitional ability, for example, one can state that a group
of agents $A$ has a strategy (a choice of actions) such that whatever
the actions by the agents outside the coalition, any computation of
the system generated by the strategy satisfies some temporal
property. A number of variations on the semantics of ATL exist: agents
may have perfect recall or be memoryless, and they may have full or
partial observability. In the case of fully observable models and
memoryless agents, the model-checking problem for ATL is polynomial in
the size of the model and the formula, while it is undecidable for
partially observable models where agents have perfect
recall~\cite{dastani+10}. Additionally, even in the simple case of
fully observable models and memoryless agents, the complexity
increases substantially if the model-checking problem takes into
account models with \emph{compact} (implicit)
representations~\cite{dastani+10}. 

In this paper, we consider an extension of perfect recall, fully
observable ATL where agents produce and consume resources. The
properties we are interested in are related to coalitional ability
under resource bounds. Instead of asking whether a group of agents has
a strategy to enforce a certain temporal property, we are interested in
whether the group has a strategy that can be executed under a certain
resource bound (\eg if the agents have at most $b_1$ units of resource
$r_1$ and $b_2$ units of resource $r_2$). Clearly, some actions may no
longer be used as part of the strategy if their cost exceeds the
bound. There are several ways in which the precise notion of the cost
of a strategy can be defined. For example, one can define it as the
maximal cost of any path (computation of the system) generated by the
strategy, where the cost of a path is the sum of resources produced
and consumed by actions on the path.  We have chosen a different
definition which says that a strategy has a cost at most $b$ if for
every path generated by the strategy, every \emph{prefix} of the path
has cost at most $b$.  This means that a strategy cannot, for example,
start with executing an action that consumes more than $b$ resources,
and then `make up' for this by executing actions that produce
enough resources to bring the total cost of the path under $b$. It is
however possible to first produce enough resources, and then execute
an action that costs more than $b$, so long as the cost of the path is
less than $b$.

There are also many choices for the precise syntax of the logic and the truth
definitions of the formulas. For example, in \cite{Bulling/Farwer:10a} several
versions are given, intuitively corresponding to considering resource bounds
both on the coalition $A$ and the rest of the agents in the system, 
considering a fixed resource endowment of $A$ in the initial state which 
affects their endowment after executing some actions, etc.  
In this paper we give a precise comparison of our logic with the variants of
${\cal L}_{RAL}$ introduced in \cite{Bulling/Farwer:10a}, and in the process
solve an open problem stated in \cite{Bulling/Farwer:10a}.
In \cite{DellaMonica//:11a,DellaMonica//:13a} different syntax and semantics 
are considered, in which the resource endowment of the whole system is taken into account when 
evaluating a statement concerning a group of agents $A$. As observed in 
\cite{Bulling/Farwer:10a}, subtle differences in truth conditions for 
resource logics result in the difference between decidability and 
undecidability of the model-checking problem.
In \cite{Bulling/Farwer:10a}, the undecidability of several versions of the
logics is proved. Recently, even more undecidability results were shown in \cite{Bulling/Goranko:13a}. The only decidable cases considered in \cite{Bulling/Farwer:10a} are an extension of Computation Tree Logic (CTL) \cite{Clarke//:86a} with resources (essentially one-agent ATL) and the version where on every path only
a fixed finite amount of resources can be produced. Similarly, \cite{DellaMonica//:11a}  gives a decidable logic, PRB-ATL (Priced Resource-Bounded ATL), where
the total amount of resources in the system has a fixed bound. The model-checking algorithm for PRB-ATL runs in time polynomial in the sizes of the
model and the formula, and exponential in the number of resources and the size 
of the representation (if in binary) of the resource bounds. In \cite{DellaMonica//:13a} an EXPTIME lower bound in the number 
resources and in the size 
of the representation (if in binary) of the resource bounds is shown. 

The structure of this paper is as follows. In sections \ref{sec:rbatl},
\ref{sec:mc}, and \ref{sec:complexity}, we introduce Resource-Bounded ATL with
production and consumption of resources, a model-checking algorithm for it, and
prove that the model-checking problem is EXPSPACE-hard. This part of the paper
extends \cite{Alechina//:14a}. In section \ref{sec:feasible} we discuss two
special cases with feasible model-checking, one of them being a generalisation
of the model-checking algorithm for (production-free) RB-ATL introduced in
\cite{Alechina//:10a} to unbounded resources. In section \ref{sec:ral} we give a
detailed comparison with the logics in \cite{Bulling/Farwer:10a} and show that
for one of them the model-checking problem is decidable, solving an open problem
stated in \cite{Bulling/Farwer:10a}.\footnote{Intuitively, the main difference between
our logic (with a decidable model-checking problem) and a version of RAL from
\cite{Bulling/Farwer:10a} where the model-checking problem is undecidable under infinite semantics (considering only infinite computations) is that in our logic, each agent always has
an option of executing an \emph{idle} action which does not consume any 
resources. This means that a finite strategy which conforms to a resource bound
and enforces a particular outcome can always be extended to an infinite 
strategy by chosing the \emph{idle} action. The model-checking problem for 
the same version of RAL but under finite semantics (considering finite 
computations) turns out also to be decidable, and a model-checking algorithm
for it is obtained as an easy adaptation of the model-checking algorithm
for our logic.}

\section{Syntax and Semantics of \atlr}
\label{sec:rbatl}

The logic \RBATL was introduced in \cite{Alechina//:10a}. Here we generalise the definitions from \cite{Alechina//:10a}
to allow for production as well as consumption of resources. To avoid confusion with the consumption-only version
of the logic from \cite{Alechina//:10a}, we refer to \RBATL with production and consumption of resources as \atlr.

Let $\Agt =\{a_1, \ldots, a_n\}$ be a set of $n$ agents,
    $\Res=\{res_1, \ldots, res_r\}$ be a set of $r$ resources, 
    $\Pi$ be a set of propositions and
    $B = \natinfty^r$ be a set of resource bounds where $\natinfty = \nat \cup \{\infty\}$.


Formulas of \atlr are defined by the following syntax 
\[
\phi, \psi ::= p \mid 
            \neg \phi \mid 
            \phi \lor \psi \mid
            \AO{A^b}{\phi} \mid
            \AG{A^b}{\phi} \mid
            \AU{A^b}{\phi}{\psi}
\]
where $p \in \Pi$ is a proposition, $A \subseteq \Agt$, and $b \in B$ is a resource bound.
Here, $\AO{A^b}{\phi}$ means that a coalition $A$ can ensure that the next state satisfies $\phi$
under resource bound $b$. $\AG{A^b}{\phi}$ means that $A$ has a strategy to
make sure that $\phi$ is always true, and the cost of this strategy is at most $b$. Similarly,
$\AU{A^b}{\phi}{\psi}$ means that $A$ has a strategy to enforce $\psi$ while maintaining the truth
of $\phi$, and the cost of this strategy is at most $b$.


We extend the definition of a concurrent game structure with resource consumption
and production.

\begin{definition}
\label{def:rbcgs}
A resource-bounded concurrent game structure (RB-CGS) is a tuple 
$M = (\Agt, \Res, S, \Pi, \pi, Act, d, c, \delta)$ where:
\begin{itemize}
\item $\Agt$ is a non-empty set of $n$ agents, $\Res$ is a non-empty set
  of $r$ resources and $S$ is a non-empty set of states;
\item $\Pi$ is a finite set of propositional variables and $\pi : \Pi
  \to \wp(S)$ is a truth assignment which associates each proposition
  in $\Pi$ with a subset of states where it is true;
\item $\Act$ is a non-empty set of actions which includes $idle$, and\sloppy
$d : S \times \Agt \to \wp(\Act)\setminus \{\emptyset\}$ is a function
  which assigns to each $s \in S$ a non-empty set of actions available
  to each agent $a \in \Agt$. For every $s \in S$ and $a \in \Agt$,
  $idle \in d(s,a)$. We denote joint actions by all agents in $\Agt$
  available at $s$ by $D(s) = d(s,a_1) \times \cdots \times d(s,a_n)$;
\item $c : S \times \Agt \times \Act \to \integer^r$ is a partial function which
maps a state $s$, an agent $a$ and an action $\alpha \in d(s,a)$ to a vector of
integers, where the integer in position $i$ indicates consumption or production
of resource $res_i$ by the action (positive value for consumption and negative
value for production). We stipulate that $c(s,a,idle)=\bar 0$ for all $s \in S$
and $a \in \Agt$, where $\bar 0 = 0^r$.
\item $\delta : S \times \Act^{|\Agt|} \to S$ is a partial function that maps 
for every $s \in S$ and joint action $\sigma \in D(s)$ to a
state resulting from executing $\sigma$ in $s$.
\end{itemize}
\end{definition}

Given a RB-CGS $M$, we denote the set of all infinite sequences of states
(infinite computations) by $S^\omega$ and the set of non-empty finite sequences
(finite computation) of states by $S^+$. For a computation $\lambda = s_0 s_1
\ldots \in S^\omega$ 
we use the notation $\lambda[i] = s_i$ 
and $\lambda[i,j] = s_i \ldots s_j$. 

Given a RB-CGS $M$ and a state $s \in S$, a \emph{joint action by a
coalition} $A \subseteq \Agt$ is a tuple $\sigma = (\sigma_a)_{a \in
  A}$ (where $\sigma_a$ is the action that agent $a$ executes as part of $\sigma$, the $a$th component of $\sigma$) such that $\sigma_a \in d(s,a)$.
The set of all joint actions for $A$ at state $s$ is denoted by
$D_A(s)$. Given a joint action by the grand coalition $\sigma \in D(s)$, $\sigma_A$ (a projection of $\sigma$ on $A$)
denotes the joint action executed by $A$ as part of $\sigma$: $\sigma_A = (\sigma_a)_{a \in A}$. The set of all possible outcomes
of a joint action $\sigma \in D_A(s)$ at state $s$ is:
\[
out(s,\sigma) = \{ s' \in S \mid \exists \sigma' \in D(s): \sigma= \sigma'_A \land 
                                           s' = \delta(s,\sigma') \}
                                           \]

In the sequel, we use the usual point-wise notation for vector comparison and
addition. In particular, $(b_1,\ldots,b_r) \leq (d_1,\ldots,d_r)$ iff $b_i \leq
d_i$ $\forall$ $i \in \{1,\ldots,r\}$, and $(b_1,\ldots,b_r) + (d_1,\ldots,d_r)
= (b_1 + d_1,\ldots,b_r + d_r)$. We assume that for any $b \in \nat$, $b \leq
\infty$ and $b + \infty$ and $\infty - b = \infty$. Given a function $f$
returning a vector, we also denote by $f_i$ the function that return the i-th
component of the vector returned by $f$.

The cost of a joint action $\sigma \in D_A(s)$ is defined as 
$cost_A(s,\sigma) = \sum_{a \in A} c(s,$ $a,\sigma_a)$ and the subscript $A$ is omitted when $A = \Agt$.

Given a RB-CGS $M$, a \emph{strategy for a coalition} $A \subseteq \Agt$ is a mapping
$F_A : S^+ \to \Act^{|A|}$ such that, for every $\lambda s \in S^+$, $F_A(\lambda s)
\in D_A(s)$. A computation $\lambda \in S^\omega$ is consistent with a strategy $F_A$ iff,
for all $i \geq 0$, $\lambda[i+1] \in out(\lambda[i],F_A(\lambda[0,i]))$. We
denote by $out(s,F_A)$ the set of all computations $\lambda$ 
starting from $s$ that are consistent with $F_A$.

Given a bound $b \in B$, a computation $\lambda \in out(s,F_A)$ is $b$-consistent
with $F_A$ iff, for every $i \geq 0$,
\[
\sum_{j=0}^{i} cost_A(\lambda[j], F_A(\lambda[0,j])) \leq b
\]
Note that this definition implies that the cost of every prefix of the computation is below $b$.

The set of all computations starting from state $s$ that are $b$-consistent with $F_A$  is denoted by
$out(s,F_A,b)$.  $F_A$ is a $b$-strategy iff $out(s,F_A) = out(s,F_A,b)$ for any state $s$.


Given a RB-CGS $M$ and a state $s$ of $M$, the truth of a \atlr formula $\phi$
with respect to $M$ and $s$ is defined inductively on the structure of $\phi$
as follows:
\begin{itemize}
 \item $M, s \models p$ iff $s \in \pi(p)$;
 
 \item $M, s \models \neg \phi$ iff $M, s \not\models \phi$;
 
 \item $M, s \models \phi \lor \psi$ iff $M, s \models \phi$ or $M, s \models 
 \psi$;
 
\item $M, s \models \AO{A^b}{\phi}$ iff $\exists$ $b$-strategy $F_A$ such that
      for all $\lambda \in out(s, F_A)$: $M, \lambda[1] \models \phi$;

\item $M, s \models \AG{A^b}{\phi}$ iff $\exists$ $b$-strategy $F_A$ such that
      for all $\lambda \in out(s, F_A)$ and $i \geq 0$: $M, \lambda[i] \models
      \phi$; and

\item $M, s \models \AU{A^b}{\phi}{\psi}$ iff $\exists$ $b$-strategy $F_A$ such
      that for all $\lambda \in out(s, F_A)$, $\exists i \geq 0$: $M, \lambda[i]
      \models \psi$ and $M, \lambda[j] \models \phi$ for all $j \in
      \{0,\ldots,i-1\}$.

\end{itemize}

Since the infinite resource bound version of \atlr modalities correspond 
to the standard ATL modalities, we will write $\AO{A^{\bar{\infty}}}{\phi}$,
$\AU{A^{\bar{\infty}}}{\phi}{\psi}$, $\AG{A^{\bar{\infty}}}{\phi}$ as 
$\AO{A}{\phi}$,$\AU{A}{\phi}{\psi}$,$\AG{A}{\phi}$, respectively.
When the context is clear, we will sometimes write $s \models \phi$ instead of $M,s \models \phi$.

Note that although we only consider infinite paths, the condition that 
the $idle$ action is always available and costs $\bar{0}$ makes the 
model-checking problem easier (we only need to find a strategy with a 
finite prefix under bound $b$ to satisfy formulas of the form $\AO{A^b}{\phi}$
and $\AU{A^b}{\phi}{\psi}$, and then the strategy can make the $idle$ choice 
forever). 

As an example of the expressivity of the logic, consider the model in
Figure~\ref{fig:example} with two agents $a_1$ and $a_2$ and two resources
$r_1$ and $r_2$.
Let us assume that $c(s_I,a_1,\alpha) = \langle -2,1 \rangle$ (action $\alpha$
produces 2 units of $r_1$ and consumes one unit of $r_2$), $c(s,a_2,\beta) = \langle 1,-1 \rangle$
and $c(s,a_1,\gamma) = \langle 5,0 \rangle$. Then agent $a_1$ on its own has
a strategy to enforce a state satisfying $p$ under resource bound of
$3$ units of $r_1$ and $1$ unit of $r_2$ ($M,s_I \models \AU{\{a_1\}^{\langle 3,1\rangle}}{\top}{p}$): 
$a_1$ has to select action $\alpha$ in 
$s_I$ which requires it to consume one unit of $r_2$ but produces two units of
$r_1$, and then action $\gamma$ in $s$ that requires $5$ units of $r_1$ which is
now within the resource bound since the previous action has produced $2$ units.
All outcomes of this strategy lead to $s'$ where $p$ holds. After this, $a_1$
has to select $idle$ forever, which does not require any resources. Any smaller
resource bound is not sufficient. However, both agents have a strategy to enforce
the same outcome under a smaller resource bound of just one unit of $r_2$
($M,s_I \models \AU{\{a_1,a_2\}^{\langle 0,1\rangle}}{\top}{p}$): agent $a_2$
needs to select $\beta$ and $a_1$ $\idle$ in $s$ until the agents have gone through the loop
between $s_I$ and $s$ four times and accumulated enough of resource $r_1$ to
enable agent $a_1$ to perform $\gamma$ in $s$.

\begin{figure}
\centering
\includegraphics[width=12cm]{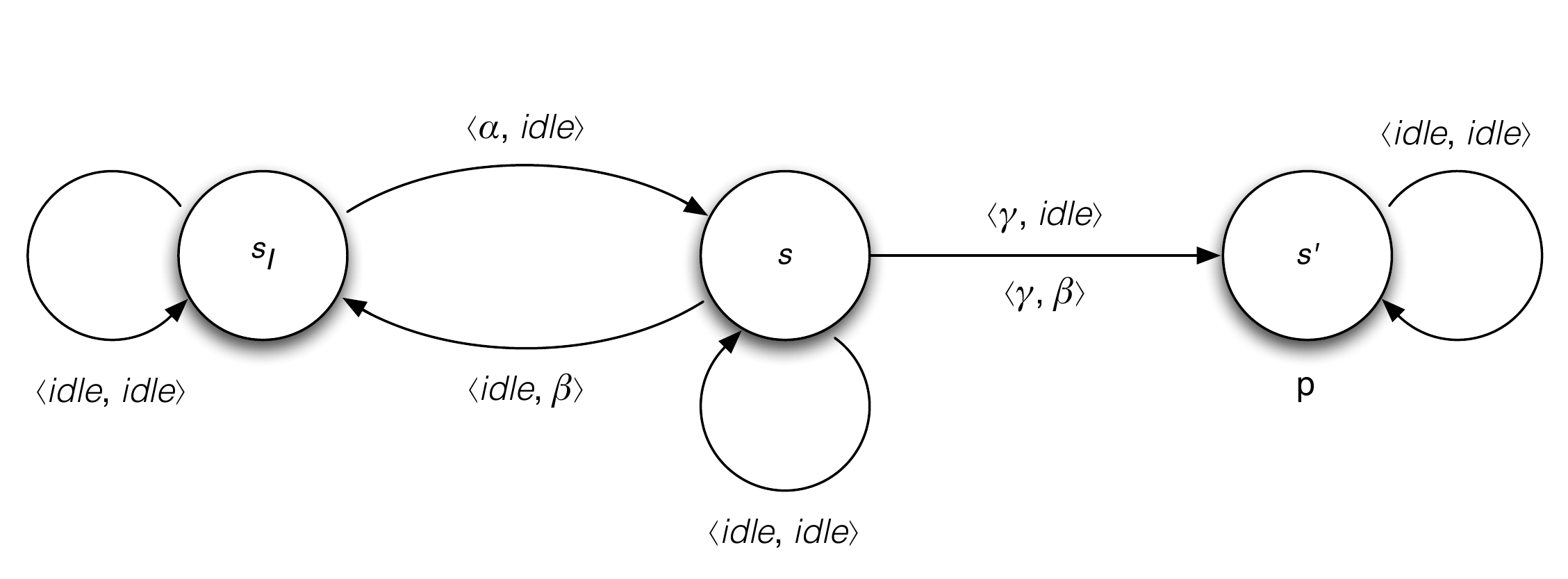}
\caption{An example with consumption and production of resources.}\label{fig:example}
\end{figure}

\section{Model Checking \atlr}
\label{sec:mc}

The model-checking problem for \atlr is the question whether, for a given 
RB-CGS structure $M$, a state $s$ in $M$ and an \atlr formula $\phi_0$, 
$M,s \models \phi_0$. In this section we prove the following theorem:

\begin{theorem}
The model-checking problem for \atlr is decidable.
\end{theorem}

To prove decidability, we give an algorithm which, given a structure $M = (\Agt, \Res, S, \Pi, \pi, Act, d, c, \delta)$ and a formula $\phi_0$, returns the set of states $[\phi_0]_M$ satisfying $\phi_0$: $[\phi_0]_M = \{s\ |\ M,s \models \phi_0\}$ (see Algorithm \ref{alg:atlr-label}).

\begin{algorithm}
\caption{Labelling $\phi_0$ }
\label{alg:atlr-label}
\begin{algorithmic}
\Function{rb$\pm$atl-label}{$M, \phi_0$}
\For{$\phi' \in Sub(\phi_0)$}
\Case{$\phi' = p,\ \neg \phi,\ \phi \wedge \psi$, $\AO{A}{\phi}$, $\AU{A}{\phi}{\psi}$, $\AG{A}{\phi}$}
\State standard, see \cite{Alur//:02a}
\EndCase
\Case{$\phi' = \AO{A^b}{\phi}$}
\State $[\phi']_M \gets Pre(A, [\phi]_M,b)$
\EndCase
\Case{$\phi' = \AU{A^b}{\phi}{\psi}$}
\State $[\phi']_M \gets \{\ s \mid s \in S \wedge$
\State $\quad \Call{until-strategy}{node_0(s,b), \AU{A^b}{\phi}{\psi} } \}$
\EndCase
\Case{$\phi' = \AG{A^b}{\phi}$}
\State $[\phi']_M \gets \{\ s \mid s \in S \wedge \Call{box-strategy}{node_0(s,b), \AG{A^b}{\phi} } \}$
\EndCase
\EndFor
\State $\mathbf{return\ } [\phi_0]_M$
\EndFunction
\end{algorithmic}
\end{algorithm}

Given $\phi_0$, we produce a set of subformulas $Sub(\phi_0)$ of $\phi_0$ 
in the usual way, however, in addition, if $\llangle A^b \rrangle \gamma \in  Sub(\phi)$, its
infinite resource version $\llangle A \rrangle \gamma$ is added to $Sub(\phi)$. 
$Sub(\phi)$ is ordered in increasing order of complexity, and the infinite resource version of each modal formula comes before the bounded version. Note that if a state $s$ is not annotated with $\llangle A \rrangle \gamma$ then $s$ cannot satisfy the bounded resource 
version $\llangle A^b \rrangle \gamma$.

We then proceed by cases. For all formulas in $Sub(\phi)$ apart from $\AO{A^b}\phi$, $\AU{A^b}{\phi}{\psi}$ and $\AG{A^b}{\phi}$ we essentially run the standard ATL model-checking algorithm \cite{Alur//:02a}. 

Labelling states with $\AO{A^b}\phi$ makes use of a function 
$Pre(A,\rho,b)$ which, given a coalition $A$, a set $\rho \subseteq S$
and a bound $b$, returns a set of states $s$ in which $A$ has a joint action $\sigma_A$ with $cost(s,\sigma_A) \leq b$ such that $out(s,\sigma_A) \subseteq \rho$.
Labelling states with $\AU{A^b}{\phi}{\psi}$ and $\AG{A^b}{\phi}$ is
more complex, and in the interests of readability we provide separate
functions: \textsc{until-strategy} for $\AU{A^b}{\phi}{\psi}$ formulas
is shown in Algorithm \ref{alg:until-strategy-multiple-resources}, and
\textsc{box-strategy} for $\AG{A^b}{\phi}$ formulas is shown in
Algorithm \ref{alg:box-strategy-multiple-resources}.

Both algorithms proceed  by depth-first and-or search of $M$. We record information about the state of the search in a search tree of nodes. A \emph{node} is a structure which consists of a state of $M$, the resources available to the agents $A$ in that state (if any), and a finite path of nodes leading to this node from the root node. Edges in the tree correspond to joint actions by all agents. Note that the resources available to the agents in a state $s$ on a path constrain the edges from the corresponding node to be those actions $\sigma_A$ where $cost(s,\sigma_A)$ is less than or equal to the available resources.
For each node $n$ in the tree, we have a function $s(n)$ which returns its
state, $p(n)$ which returns the nodes on the path and $e_i(n)$ which returns the
resource availability on the $i$-th resource in $s(n)$ as a result of following
$p(n)$. The function $\mathit{node}_0(s,b)$ returns the root node, i.e., a node
$n_0$ such that $s(n_0) = s$, $p(n_0) = [\ ]$ and $e_i(n_0) = b_i$ for all
resources $i$. The function $\mathit{node}(n, \sigma, s')$ returns a node $n'$
where $s(n') = s'$, $p(n') = [p(n) \cdot n]$ and for all resources $i$, $e_i(n')
= e_i(n) - cost_i(\sigma)$.

Both \textsc{until-strategy} and \textsc{box-strategy} take a search tree node $n$ and a formula $\phi' \in Sub(\phi_0)$ as input, and have similar structure. They first check if the infinite resource version of $\phi'$ is false in the state represented by node $n$, $s(n)$. If so, they return false immediately, terminating search of the current branch of the search tree. \textsc{until-strategy} also returns true if the second argument $\psi$ of $\phi'$ is true in $s(n)$. Both
\textsc{until-strategy} and \textsc{box-strategy} check whether the state $s(n)$ has been encountered before on $p(n)$, i.e., $p(n)$ ends in a loop. In the case of \textsc{until-strategy}, if the loop is unproductive (i.e., resource availability has not increased since the previous occurrence of $s(n)$ on the path), then the loop is not necessary for a successful strategy, and search on this branch is terminated. If on the other hand the loop strictly increases the availability of at least one resource $i$ and does not decrease the availability of other resources, then 
$e_i(n)$ is replaced with $\infty$ (as a shorthand denoting that any finite amount of $i$ can be produced by repeating the loop sufficiently many times). If all
resource values have been replaced by $\infty$, \textsc{until-strategy} returns true, since the branch satisfies the infinite resource version
$\AU{A}{\phi}{\psi}$ of $\phi'$, and an arbitrary amount of any resource can
be accumulated along the path. For \textsc{box-strategy} the loop check is 
slightly different. If the loop decreases the amount of at least one resource without increasing the availability of any other resource, it cannot form part of
a successful strategy, and the search terminates returning false. If a 
non-decreasing loop is found, then it is possible to maintain the invariant 
formula $\phi$ forever without expending any resources, and the search 
terminates returning true.  

If the none of the if statements evaluates to true, then, in both \textsc{until-strategy} and \textsc{box-strategy}, search continues by considering each action available at $s(n)$ in turn. For each action $\sigma \in ActA$, the algorithm checks whether a recursive call of the algorithm returns true in all outcome states of $\sigma$  (i.e., $\sigma$ is part of a successful strategy). If such a $\sigma$ is found, the algorithm returns true. Otherwise the algorithm returns false. Note that the argument $\phi'$ is passed through the recursive calls unchanged: information about the resources available to the agents in $s(n)$ as a result of following $p(n)$ is encoded in the search nodes.

\begin{algorithm}[h]
\caption{Labelling $\AU{A^b}{\phi}{\psi}$ }
\label{alg:until-strategy-multiple-resources}

\begin{algorithmic} 
\Function{until-strategy}{$n, \AU{A^b}{\phi}{\psi} $}
\If{$s(n) \not\models \AU{A}{\phi}{\psi} $}
\State $\mathbf{return}\ \mathit{false}$ 
\EndIf
\If{$\exists n' \in p(n): s(n') = s(n) \wedge (\forall j: e_j(n') \geq e_j(n))$}
\State $\mathbf{return}\ \mathit{false}$ 
\EndIf
\For{$i \in \{i \in \Res \mid \exists n' \in p(n): s(n') = s(n) \wedge (\forall j: e_j(n') \leq e_j(n)) \wedge e_i(n') < e_i(n)\}$}
\State $e_i(n) \gets \infty$
\EndFor
\If{$s(n) \models \psi $}
\State $\mathbf{return}\ \mathit{true}$ 
\EndIf
\If{$e(n) = \bar \infty$}
\State $\mathbf{return}\ \mathit{true}$ 
\EndIf
\State $ActA \gets \{ \sigma \in D_A(s(n)) \mid cost(s(n),\sigma) \leq e(n) \}$
\For{$\sigma \in ActA $}
\State $O \gets out(s(n),\sigma)$ 
\State $\mathit{strat} \gets \mathit{true}$
\For{$s' \in O$}
\State $\mathit{strat} \gets \mathit{strat} \wedge $
\State $\quad \Call{until-strategy}{node(n,\sigma,s'), \AU{A^b}{\phi}{\psi} }$
\EndFor
\If{$ \mathit{strat}$}
\State $\mathbf{return}\ \mathit{true}$
\EndIf
\EndFor
\State $\mathbf{return}\ \mathit{false}$
\EndFunction
\end{algorithmic}
\end{algorithm}

\begin{algorithm}[h]
\caption{Labelling $\AG{A^b}{\phi}$ }
\label{alg:box-strategy-multiple-resources}
\begin{algorithmic} 
\Function{box-strategy}{$n$, $\AG{A^b}{\phi}$}
\If{$s(n) \not\models \AG{A}{\phi}$}
\State $\mathbf{return}\ \mathit{false}$
\EndIf
\If{$\exists n' \in p(n): s(n') = s(n) \wedge (\forall j: e_j(n') \geq e_j(n)) \wedge (\exists j: e_j(n') > e_j(n))$}
\State $\mathbf{return}\ \mathit{false}$
\EndIf
\If{$\exists n' \in p(n): s(n') = s(n) \wedge (\forall j: e_j(n') \leq e_j(n))$}
\State $\mathbf{return}\ \mathit{true}$
\EndIf
\State $ActA \gets \{ \sigma \in D_A(s(n)) \mid cost(s(n),\sigma) \leq e(n) \}$
\For{$\sigma \in ActA $}
\State $O \gets out(s(n),\sigma)$
\State $\mathit{strat} \gets \mathit{true}$
\For{$s' \in O$}
\State $\mathit{strat} \gets \mathit{strat} \wedge$
\State $\quad \Call{box-strategy}{node(n,\sigma,s'), \AG{A^b}{\phi} }$
\EndFor
\If{$ \mathit{strat}$}
\State $\mathbf{return}\ \mathit{true}$
\EndIf
\EndFor
\State $\mathbf{return}\ \mathit{false}$
\EndFunction
\end{algorithmic}
\end{algorithm}

\begin{lemma}
Algorithm \ref{alg:atlr-label} terminates.
\end{lemma}
\begin{proof} 
All the cases in Algorithm~\ref{alg:atlr-label} apart from $\AU{A^b}{\phi}{\psi}$ and $\AG{A^b}{\phi}$
can be computed in time polynomial in $|M|$ and $|\phi|$. The cases for $\AU{A^b}{\phi}{\psi}$ and $\AG{A^b}{\phi}$
involve calling the \textsc{until-strategy} and \textsc{box-strategy} procedures,
respectively, for every state in $S$. We want to show that there is no
infinite sequence of calls to \textsc{until-strategy} or \textsc{box-strategy}.
Assume to the contrary that $n_1,n_2,\ldots$ is
an infinite sequence of nodes in an infinite sequence of recursive calls to \textsc{until-strategy} or 
\textsc{box-strategy}.
Then, since the set of states is finite, there is an infinite
subsequence $n_{i_1},n_{i_2},\ldots$ of $n_1,n_2,\ldots$ such that for all $j$, $s(n_{i_j}) = s$
for some state $s$ (the state is the same for all the nodes in the subsequence). 
We show that then there is an infinite subsequence $n'_1,n'_2,\ldots$ of $n_{i_1},n_{i_2},\ldots$
such that for $k < j$, $e(n'_k) \leq e(n'_j)$. Note that since all nodes have the
same state, this implies that both \textsc{until-strategy} or \textsc{box-strategy} will return 
after finitely many steps: a contradiction.
The proof is very similar to the proof of Lemma f in \cite[p.70]{Reisig.1985} and proceeds by
induction on the number of resources $r$. For $r=1$, since $e(n)$ is always positive, the claim
is immediate. Assume the lemma holds for $r$ and let us show it for $r+1$. Then there is an infinite subsequence
$m'_1,m'_2,\ldots$ of $n_{i_1},n_{i_2},\ldots$ where for all resources $i \in \{1,\ldots,r\}$ 
$e_i(m'_k) \leq e_i(m'_j)$ for $k < j$. Clearly there are two nodes $m'_{j_1}$
and $m'_{j_2}$ in this sequence such that $e_{r+1}(m'_{j_1}) \leq e_{r+1}(m'_{j_2})$ (since there are only finitely many positive integers which are smaller
than $e_{r+1}(m'_1)$). Hence $e(m'_{j_1}) \leq e(m'_{j_2})$ and 
the sequence of calls would terminate in $m'_{j_2}$, a contradiction.
\end{proof}

Before we prove correctness of \textsc{until-strategy} and \textsc{box-strategy}, we need some auxiliary
notions. Let $n$ be a node where one of the procedures returns true. We will refer to $tree(n)$ as the
tree representing the successful call to the procedure. In particular, if the procedure returns true
before any recursive calls are made, then $tree(n)=n$. Otherwise the procedure returns true because there is
an action $\alpha \in Act_A$ such that for all $s' \in out(s(n),\alpha)$ the procedure returns true in
$n' = node(n,\alpha,s')$. In this case, $tree(n)$ has $n$ as its root and trees $tree(n')$ are the children of $n$.
We refer to the action $\alpha$ as $n_{act}$ (the action that generates the children of $n$). For 
the sake of uniformity, if $tree(n)=n$ then we set $n_{act}$ to be $idle$.
Such a tree corresponds to a strategy $F$ where for each path $n \cdots m$ from the root $n$ to a node $m$ 
in $tree(n)$, $F(s(n) \cdots s(m)) = m_{act}$.
\sloppy

A strategy $F$ for satisfying $\AU{A^{b}}{\phi}{\psi}$ is ${\cal U}$-economical for a node $n$ if, intuitively, no path generated by
it contains a loop that does not increase any resource. A strategy is 
$\Box$-economical for a node $n$ if, intuitively, no path generated by
it contains a loop that decreases some resources and does not increase any 
other resources. Formally, a strategy $F$ is ${\cal U}$-economical for $n$ if  
\begin{itemize}
 \item $F$ satisfies $\AU{A^{e(n)}}{\phi}{\psi}$ at $s(n)$, i.e., $F$ is a
       $e(n)$-strategy and $\forall \lambda \in out(s(n),F)$, $\exists i \geq 0:
       \lambda[i] \models \psi$ and $\lambda[j] \models \phi$ for all $j \in
       \{0,\ldots,i\}$
 \item The path $p(n) \cdot n$ is already ${\cal U}$-economical, i.e.,
 $\forall n' \in p(n)\cdot n, n'' \in p(n'): s(n'') = s(n') \Rightarrow e(n'') \not\geq e(n')$; 
 \item Every state is reached by $F$ ${\cal U}$-economically, i.e., for each
       computation $s_0s_1\ldots s_k \ldots \in out(s(n),$ $F)$ and $j < k \leq
       i$ where $i$ is the first index such that $s_i$ satisfies $\psi$, $s_j =
       s_k \Rightarrow cost(s_j\ldots s_k) \not\geq \bar{0}$ with
       $cost(s_j\ldots s_k) = \sum_{l=j,\dots,k-1}cost($
       $\lambda[l],F(\lambda[0,l]))$; and
 \item Every state is reached by $F$ ${\cal U}$-economically with respect to the path $p(n)$, i.e.,
 for every computation $s_0s_1\ldots s_k \ldots \in out(s(n),F)$,
 $\forall n' \in p(n): s(n') = s_k \Rightarrow e(n') \not\geq e(n) - cost(s_0\ldots s_k)$ 
\end{itemize}

A strategy $F$ is $\Box$-economical if:
\begin{itemize}
 \item $F$ satisfies $\AG{A^{e(n)}}{\phi}$ at $s(n)$, i.e., $F$ is a
       $e(n)$-strategy and $\forall \lambda \in out(s(n),F)$, $\forall i \geq 0:
       \lambda[i] \models \phi$;
 \item The path $p(n) \cdot n$ is already $\Box$-economical, i.e.,
 $\forall n' \in p(n)\cdot n, n'' \in p(n'): s(n'') = s(n') \Rightarrow e(n'') \not> e(n')$;
 \item Every state is reached by $F$ $\Box$-economically, i.e.,
 for each computation $s_0s_1\ldots s_k \ldots \in out(s(n),$ $F)$ 
 $\forall j < k: s_j = s_k \Rightarrow cost(s_j\ldots s_k) \not> \bar{0}$;
 \item Every state is reached by $F$ $\Box$-economically with respect to the path $p(n)$, i.e.,
 for every computation $s_0s_1\ldots s_k \ldots \in out(s(n),F)$,
 $\forall n' \in p(n): s(n') = s_k \Rightarrow e(n') \not> e(n) - cost(s_0\ldots s_k)$.
\end{itemize}

Note that any strategy $F$ satisfying 
$\AU{A^{e(n)}}{\phi}{\psi}$ ($\AG{A^{e(n)}}{\phi}$) at $s(n)$ can be converted to an economical one
by eliminating unproductive loops:
\begin{proposition}
There is a strategy to satisfy $\AU{A^{e(n)}}{\phi}{\psi}$
($\AG{A^{e(n)}}{\phi}$) at $s(n)$ iff there is an economical strategy to satisfy
$\AU{A^{e(n)}}{\phi}{\psi}$ ($\AG{A^{e(n)}}{\phi}$) at $s(n)$.
\end{proposition}

Next we prove correctness of \textsc{until-strategy}. The next lemma essentially
shows that replacing a resource value with $\infty$ in
Algorithm~\ref{alg:until-strategy-multiple-resources} is harmless. For the
inductive step of the proof, we need the following notion. Given a tree
$tree(n)$, we call its pruning, denoted as $prune(tree(n), m_1,\ldots,m_k)$, the
tree obtained by removing all children of some nodes $m_1,\ldots, m_k$ that have
only leaves as children in $tree(n)$.

\begin{lemma}
Let $n = node_0(s,b)$ be a node where $\textsc{until-strategy}$ returns true.
Let $f$ be a function that for each leaf $n'$ of $tree(n)$ returns $f(n') \in
\nat^r$ such that $f_i(n') = e_i(n')$ if $e_i(n') \not = \infty$ ($f_i(n')$ can
be any natural number if $e_i(n') = \infty$). Then, 
there is a strategy $F$ such that for every leaf $n'$ of the tree $tree(n)$
induced by $F$, $e(n') \geq f(n')$ holds.
\end{lemma}

\begin{proof} 
By induction on the structure of $tree(n)$.

\begin{description}
  \item[Base Case:]
        Let $tree(n)$ contain only its root. The proof is obvious for any strategy. 
  \item[Inductive Step:]
        Let us consider a pruning $T$ of $tree(n)$. By the induction
        hypothesis, any tree $T'$ that has a less complex structure than $T$ has a strategy to generate at least $f(n')
        \in \nat^r \leq e(n')$ for all leaves $n'$ of $T'$.
        \begin{figure}[h]
        \centering
        \def\svgwidth{0.5\textwidth} 
        {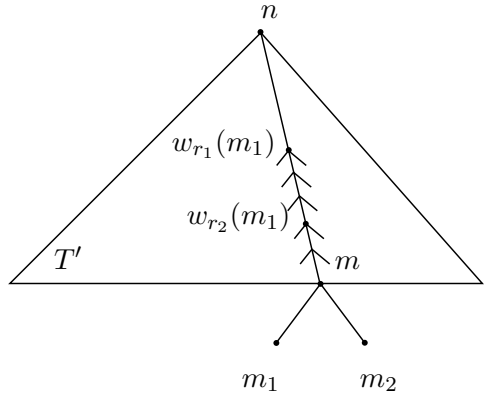}
        \caption{Tree $T$ and $T'=prune(T,m)$.}
        \label{fig:smaller-tree}
        \end{figure}

        In the following, given nodes $n, n_1, \ldots, n_k$, we denote by
        $n(n_1,\ldots,n_k)$ the depth-1 tree which has $n$ as its root and
        $n_1,\ldots,n_k$ as the immediate leaves of $n$.

        Let $m(m_1,\ldots,m_k)$ be an arbitrary depth-1 sub-tree of $T$ (see
        Figure \ref{fig:smaller-tree}). By removing $m_1,\ldots,m_k$ from $T$,
        we obtain a pruning $T'$ of $T$.

        Let $n \cdots m \cdot m_i$ be a path in $T$ from the root $n$ to one of
        the leaves $m_i$. For each resource $r$ the availability of which turns
        to $\infty$ at $m_i$, there must be a node, denoted by $w_r(m_i)$, in
        the path $n \cdots m \cdot m_i$ which is used to turn the availability
        of $r$ to $\infty$ at $m_i$, that is, $w_r(m_i)$ is such that
        $s(w_r(m_i)) = s(m_i)$, $e_i(w_r(m_i)) \leq e_i(m_i)$ for each $i$, and
        $e_r(w_r(m_i)) < e_r(m_i)$. We may repeat the path from $w_r(m_i)$ to
        $m_i$ several times to generate enough resource availability for $r$. We
        call the path from $w_r(m_i)$ to $m_i$ together with all the immediate
        child nodes of those along the path the column graph from $w_r(m_i)$ to
        $m_i$. Each time, an amount of $g_r = e_r(m) - cost_r(m_{act}) -
        e_r(w(m_i))$ is generated. Then, the minimal number of times to repeat the
        path from $w(m_i)$ to $m_i$ is $h_r(m_i) = \ceil{\frac{f_r(m_i) - (e_r(m)
        - cost_r(m_{act}))}{g_r}}$.
        
        Note that we need to repeat at each $m_i$ for each resource $r$ the path
        from $w_r(m_i)$ to $m_i$ $h_r(m_i)$ times. To record the number of times the
        path has been repeated, we attach to each $m_i$ a counter $\hat h_r(m_i)$ for
        each $r$ and write the new node of $m_i$ as $m_i^{\hat h(m_i)}$.
        \begin{figure}[h]
        \centering
        \def\svgwidth{0.8\textwidth}
        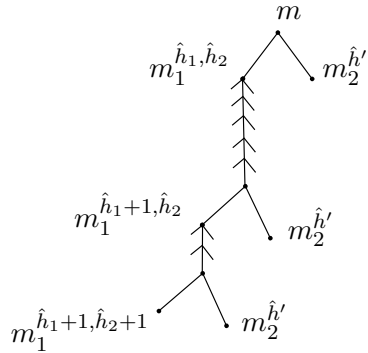
        \caption{Repeating steps to generate resources.}
        \label{fig:repeat}
        \end{figure}

        Initially, $\hat h_r(m_i) =0$ for all $r$ and for all nodes $m_i$. A
        step (see Figure \ref{fig:repeat}) of the repetition is done as follows:
        let $m_i^{\hat h(m_i)}$ be some node such that $\hat h_r(m_i) <
        h_r(m_i)$. Let $m_j^{\hat h(m_j)}$ be the sibling of $m_i^{\hat h(m_i)}$
        ($j\not=i)$. We extend from $m_i^{\hat h(m_i)}$ the column graph from
        $w_r(m_i)$ to $m_i$; each new $m_j$ ($j\not=i$) is annotated with $\hat
        h(m_j)$ (same as before) and the new $m_i$ is annotated with $\hat
        h(m_i)$ except that $\hat h_r(m_i)$ is increased by 1. We repeat the
        above step until $\hat h_r(m_i) = h_r(m_i)$ (it must terminate due to
        the fact that $h_r(m_i) < \infty$ for all $r$ and $m_i$).

        At the end, we obtain a tree where all leaves $m_i^{\hat h(m_i)}$ have
        $\hat h_r(m_i) = h_r(m_i)$ for all $r$, hence the availability of $r$ is
        at least $f_r$. Let $E(m)$ be the extended tree from $m$.
        
        Let $F_{T'}$ be the strategy generated by $T'$. We extend $F_{T'}$ with
        $E(m)$ for every occurrence of $m$ in $F_{T'}$ and denote this extended
        strategy $F^E_{T'}$. For all leaves $m'$ in $E(m)$ other than $m_i$, let
        $sub(T,m')$ be some sub-tree of $T$ staring from $m'$. Then, we extend
        $F_{T'}^E$ with $sub(T,m')$ for every occurrence of $m'$ in $F^E_{T'}$.
        We finally obtain a tree $F_T$ which satisfies the condition that all
        leaves $l$ have resource availability of at least $f(l)$.
\end{description}
\end{proof}
    
\begin{corollary}
If $\Call{until-strategy}{node_0(s,b),\AU{A^b}{\phi}{\psi}}$ returns true 
then $s \models \AU{A^{b}}{\phi}{\psi}$.
\end{corollary}

\begin{lemma}
If $\Call{until-strategy}{n,\AU{A^b}{\phi}{\psi}}$ returns false, then them
there is no strategy satisfying $\AU{A^{e(n)}}{\phi}{\psi}$ from $s(n)$ that is
${\cal U}$-economical for $n$.
\end{lemma}

\begin{proof}
We prove the lemma by induction on the height in the recursion tree of
$\Call{until-strategy}{\,}$ calls.
\begin{description}
  \item [Base Case:] If \emph{false} is returned by the first if-statement, then
        $s(n)\not\models \AU{A}{\phi}{\psi}$; this also means
        there is no strategy satisfying
        $\AU{A^{e(n)}}{\phi}{\psi}$ from $s(n)$.
  
        If \emph{false} is returned by the second if-statement, then any
        strategy satisfying $\AU{A^{e(n)}}{\phi}{\psi}$ from $s(n)$ is
        not economical.
        
  \item [Inductive Step:] If \emph{false} is not returned by the first two
        if-statements, then, for all actions $\sigma \in ActA$, there exists $s'
        \in out(s(n),\sigma)$ such that
        $\Call{until-strategy}{n',\AU{A^b}{\phi}{\psi}}$ (where
        $n'=node(n,\sigma,s')$) returns false. By induction hypothesis, there is
        no strategy satisfying $\AU{A^{e(n')}}{\phi}{\psi}$ from $s(n')$ that is
        $\cal{U}$-economical for $n'$. Assume to the contrary that there is an
        economical strategy satisfying $\AU{A^{e(n)}}{\phi}{\psi}$ from $s(n)$.
        Let $\sigma = F(s(n))$, then $\sigma \in ActA$. Obviously, for all $s'
        \in out(s(n),\sigma)$, $F'(\lambda) = F(s(n)\lambda)$ is an economical
        strategy from $n' = node(n,\sigma,s')$. This is a contradiction; hence,
        there is no strategy satisfying $\AU{A^{e(n)}}{\phi}{\psi}$ from $s(n)$
        that is $\cal{U}$-economical for $n$.
\end{description}
\end{proof}
\begin{corollary}
If $\Call{until-strategy}{node_0(s,b),\AU{A^b}{\phi}{\psi}}$ returns false 
then $s \not\models \AU{A^{b}}{\phi}{\psi}$.
\end{corollary}

Now we turn to Algorithm \ref {alg:box-strategy-multiple-resources} for
labelling states with $\AG{A^b}{\phi}$. First we show its soundness.

\begin{lemma}
Let $n = node_0(s,b)$.
If $\Call{box-strategy}{n,\AG{A^b}{\phi}}$ returns true then $s(n) \models \AG{A^b}\phi$.
\end{lemma}

\begin{proof}
Recall that, for each node $m$ in $tree(n)$, we denote by $sub(tree(n),m)$ the sub-tree
of $tree(m)$ rooted at $m$. 
        \begin{figure}[h]
        \centering
        \def\svgwidth{0.5\textwidth}
        {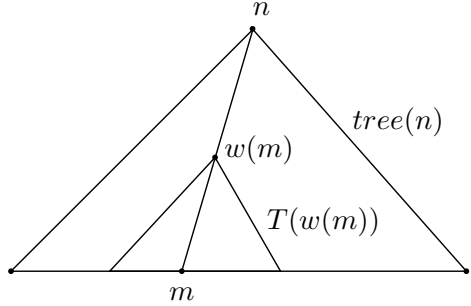}
        \caption{$w(m)$ of $m$ in $tree(n)$.}
        \label{fig:box-tree}
        \end{figure}
For each leaf $m$ of $tree(n)$,
let $w(m)$ denote one of the nodes in $p(m)$ such that $s(w(m)) = s(m)$ and
$e(w(m)) \leq e(m)$  (see Figure \ref{fig:box-tree}).

Let us expand $tree(n)$ as follows:
\begin{itemize}
\item $T^0$ is $tree(n)$;
\item $T^{i+1}$ is $T^{i}$ where all its leaves $m$ are replaced by
      $sub(tree(n),w(m))$ (see Figure \ref{fig:box-strategy}).
\end{itemize}
        \begin{figure}[h]
        \centering
        \def\svgwidth{\textwidth}
        {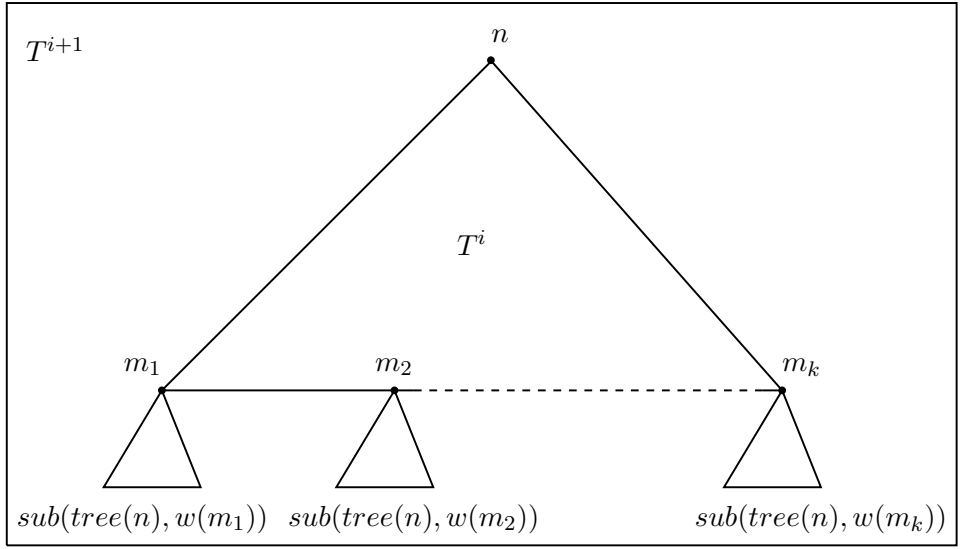}
        \caption{One step in constructing the strategy.}
        \label{fig:box-strategy}
        \end{figure}

Let $T = \lim_{i \to \infty}T^i$, then $T$ is a strategy for $\AG{A^b}\phi$.
 
\end{proof}

\begin{lemma}
If $\Call{box-strategy}{n,\AG{A^b}{\phi}}$ returns false, then there is no
strategy satisfying $\AG{A^{e(n)}}{\phi}$ from $s(n)$ that is $\Box$-economical
for $n$.
\end{lemma}

\begin{proof}
We prove the lemma by induction on the height in the recursion tree of
$\Call{box-strategy}{\,}$ calls.
\begin{description}
  \item [Base Case:] If \emph{false} is returned by the first if-statement, then
        $s(n)\not\models \AG{A}{\phi}$; this also means
        there is no strategy satisfying $\AG{A^{e(n)}}{\phi}$ at $s(n)$.
        
        If \emph{false} is returned by the second if-statement, then 
        any strategy satisfying $\AG{A^{e(n)}}{\phi}$ at $s(n)$ is not $\Box$-economical.
  \item [Inductive Step:] If \emph{false} is not returned by the first two
        if-statements, for all actions $\sigma \in ActA$, there exists $s' \in
        out(s(n),\sigma)$ such that $\Call{box-strategy}{n',\AG{A^b}{\phi}}$
        (where $n'=node(n,\sigma,s')$) returns false. Assume to the contrary
        that there is a strategy $F$ satisfying $\AG{A^{e(n)}}{\phi}$ from
        $s(n)$ that is $\Box$-economical for $n$. Let $\sigma = F(s(n))$, then
        $\sigma \in ActA$. Obviously, for all $s' \in out(s(n),\sigma)$,
        $F'(\lambda) = F(s(n)\lambda)$ is a strategy $\Box$-economical for $n' =
        node(n,\sigma,s')$. This is a contradiction; hence, there is no strategy
        satisfying $\AG{A^{e(n)}}{\phi}$ from $s(n)$ that is $\Box$-economical
        for $n$. \hfill $\Box$
\end{description}
\end{proof}
Then, we have the following result directly:
\begin{corollary}
If $\Call{box-strategy}{node_0(s,b),\AG{A^b}{\phi}}$ returns false 
then $s \not\models \AG{A^b}{\phi}$.
\end{corollary}

\section{Lower Bound}
\label{sec:complexity}

In this section we show that the lower bound for the model-checking problem for
\atlr is EXPSPACE, by reducing from the reachability problem of Petri Nets. Note
that the exact complexity of the reachability problem of Petri Nets is still an 
open question (although it is known to be decidable and EXPSPACE-hard, \cite{Reisig.1985}).
The exact complexity of the \atlr model-checking problem is also unknown. Note that 
an upper bound for the \atlr model-checking problem would also be an upper bound
for the reachability problem of Petri Nets due to the reduction below. 
Even an Ackermannian upper bound for this problem is still open \cite{Leroux:13b}.
This suggests that determining an upper bound for the \atlr model-checking problem is also a hard problem. 

A Petri net is a tuple $N = (P,T,W,M)$ where:
\begin{itemize}
 \item $P$ is a finite set of places;
 \item $T$ is a finite set of transitions;
 \item $W : P \times T \cup T \times P \to \nat$ is a weighting function; and
 \item $M : P \to \nat$ is an initial marking.
\end{itemize}

A transition $t \in T$ is $M$-enabled iff $W(r,t) \leq M(r)$ for all $r \in P$. The
result of performing $t$ is a marking $M'$ where $M'(r) = M(r) - W(r,t) +
W(t,r)$, denoted as $M \PNtran{t} M'$.

A marking $M'$ is reachable from $M$ iff there exists a sequence
\[ M_0 \PNtran{t_1} M_1 \PNtran{t_2} \ldots \PNtran{t_n} M_n\]
where $M_0 = M$ and $n \geq 0$ such that $M_n \geq M'$ (where $M \geq M'$ iff
$M(r) \geq M'(r)$ for all $r \in P$). 
It is known that the lower bound for the complexity of this version of the reachability problem (with $M_n \geq M'$
rather than $M_n = M'$) is EXPSPACE \cite[p.73]{Reisig.1985}. 

We present a reduction from an instance of the reachability problem of Petri
Nets to an instance of the model-checking problem of \atlr.

Given a net $N = (P,T,W,M)$ and a marking $M'$, we construct a RB-CGS $I_{N,M'}
= (\{1\}, P, S, \{p\}, \pi, Act, d, c, \delta)$ where:
\begin{figure}[h]
\centering
\def\svgwidth{0.8\textwidth}
{
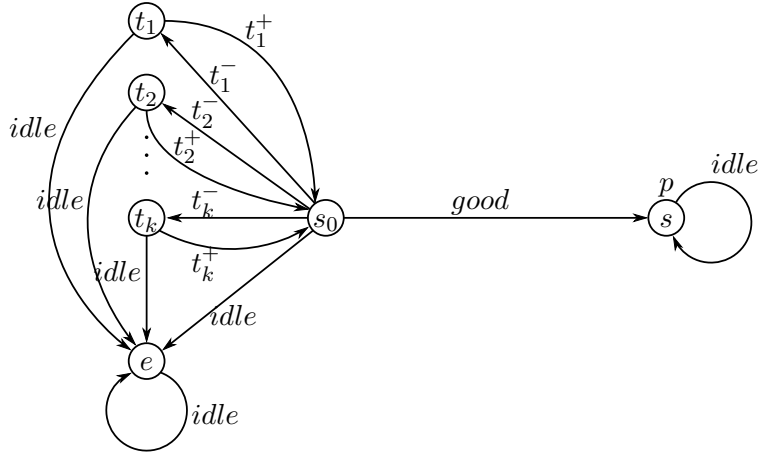}
\caption{Structure $I_{N,M'}$.}
\label{fig:reduction-structure}
\end{figure}
\begin{itemize}
 \item $S = \{ s_0 \} \cup T \cup \{s,e\}$;
 \item $\pi(p) = \{s\}$;
 \item $Act = \{idle,good\} \cup \{ t^-, t^+ \mid t \in T\}$;
 \item $d(s_0,1) = \{idle,good\} \cup \{ t^- \mid t \in T\}$;
 \item $d(s,1) = d(e,1) = \{idle\}$;
 \item $d(t,1) = \{ idle, t^+\}$ for all $t \in T$;
 \item $c(x,idle) = \bar 0$ for all $x\in S$; 
 \item $c(s_0,good) = M'$;
 \item $c_r(s_0,t^-) = W(r,t)$ for all $r \in P$;
 \item $c_r(s_0,t^+) = -W(r,i)$ for all $r \in P$;
 \item $\delta(x, idle) = e$ for $x \in S \setminus \{s\}$;
 \item $\delta(s, idle) = s$;
 \item $\delta(s_0, good) = s$;
 \item $\delta(s_0, t^-) = t$ for all $t \in T$;
 \item $\delta(t, t^+) = s_0$  for all $t \in T$.
\end{itemize}
The following is straightforward:
\begin{lemma}
 Given a net $N = (P,T,W,M)$ and a marking $M'$, $M'$ is reachable from $M$ iff
 $I_{N,M'},s_0 \models \AU{1^{M}}{\top}{p}$.
\end{lemma}
 \begin{proof} $(\Rightarrow)$:
 Assume that $M'$ is reachable from $M$, then there exists a sequence
  \[ M_0 \PNtran{t_1} M_1 \PNtran{t_2} \ldots \PNtran{t_n} M_n\]
  where $M_0 = M$ and $n \geq 0$ such that $M_n \geq M'$.
 
  Then, we consider the following strategy $F$ for agent $1$:
  \begin{itemize}
    \item $F(s_0) = t_1^-$, note that $M \geq c(s_0,t_1^-)$, additionally $\delta(s_0,t_1^-)
          = t_1$;
    \item $F(s_0t_1) = t_1^+$, note that $M - (c(s_0,t_1^-) + c(t_1,t_1^+)) = M_1 \geq
          \bar 0$, additionally $\delta(t_1,t_1^+) = s_0$;
    \item $F(s_0t_1s_0) = t_2^-$, note that $c(s_0,t_2^-) \leq M_1$, $M - (c(s_0,t_1^-) +
          c(t_1,t_1^+) + c(s_0,t_2^-)) = M_1 - c(s_0,t_2^-) \geq \bar 0$, additionally $\delta(s_0,t_2^-) = t_2$;
    \item $F(s_0t_1s_0t_2) = t_2^+$, note that $M - (c(s_0,t_1^-) + c(t_1,t_1^+) +
          c(s_0,t_2^-) + c(t_2,t_2^+)) = M_2 \geq \bar 0$, additionally $\delta(t_2,t_2^+) =
          s_0$;
 
          \hspace{0.05cm}\vdots
 
    \item $F(s_0t_1s_0t_2\ldots s_0t_n) = t_n^+$, note that $M - (c(s_0,t_1^-) +
          c(t_1,t_1^+) + c(s_0,t_2^-) + c(t_2,t_2^+) + \ldots + c(s_0,t_n^-) + c(t_n,t_n^+)) = M_n
          \geq M' \geq \bar 0$, additionally $\delta(t_n,t_n^+) = s_0$;
    \item $F(s_0t_1s_0t_2\ldots s_0t_ns_0) = good$, note that $c(s_0,good) = M'$, $M
          - (c(s_0,t_1^-) + c(t_1,t_1^+) + c(s_0,t_2^-) + c(t_2,t_2^+) + \ldots + c(s_0,t_n^-) +
          c(t_n,t_n^+) + c(s_0,good)) = M_n - M' \geq \bar 0$, additionally $\delta(s_0,good) =
          s$;
  \end{itemize}
  Since $s \models p$, it is straightforward that $F$ is a strategy satisfying
  $\AU{1^M}{\top}{p}$ from $s_0$.
 
 $(\Leftarrow)$:
 Assume that $s_0 \models \AU{1^M}{\top}{p}$, then there
                       exists a strategy $F$ which satisfies $\AU{1^M}{\top}{p}$
                       from $s_0$.
 
   Since there is only one agent, $out(s_0,F)$ contains a single path $s_0\ldots s \dots$.
   Obviously, $e$ cannot be visited on the prefix $s_0 \ldots s$; hence $s_0 \ldots s$
   must have the form $s_0 t_1 s_0 t_2 \ldots t_n s_0 s$ for some $t_1,\ldots,t_n \in T$.
   Furthermore,
   \begin{itemize}
   \item $F(s_0) = t_1^-$, $c(s_0,t_1^-) \leq M$,
   \item $F(s_0t_1) = t_1^+$, $c(s_0,t_1^-) + c(t_1,t_1^+) \leq M$,
   
   \hspace{0.05cm}\vdots
   
   \item $F(s_0t_1\ldots t_{n-1}s_0) = t_n^-$, $c(s_0,t_1^-) + c(t_1,t_1^+) + \ldots + c(s_0,t_n^-) \leq M$,
   \item $F(s_0t_1\ldots t_{n-1}s_0t_n)= t_n^+$, $c(s_0,t_1^-) + c(t_1,t_1^+) + \ldots + c(s_0,t_n^-) + c(t_n,t_n^+) \leq M$, and
   \item $F(s_0t_1\ldots t_{n-1}s_0t_ns_0)= good$, $c(s_0,t_1^-) + c(t_1,t_1^+) + \ldots + c(s_0,t_n^-) + c(t_n,t_n^+) + M' \leq M$.
   \end{itemize}
   Therefore,
   \begin{itemize}
   \item $t_1$ is $M$-enabled, let $M_1 = M - (c(s_0,t_1^-) + c(t_1,t_1^+))$,
   \item $t_2$ is $M_1$-enabled, let $M_2 = M_1 - (c(s_0,t_2^-) + c(t_2,t_2^+)) = M - (c(s_0,t_1^-) + c(t_1,t_1^+) + c(s_0,t_2^-) + c(t_2,t_2^+))$,
   
   \hspace{0.05cm}\vdots
   
   \item $t_n$ is $M_{n-1}$-enabled, let $M_n = M_{n-1} - (c(s_0,t_n^-) + c(t_n,t_n^+)) = M - (c(s_0,t_1^-) + c(t_1,t_1^+) + \ldots + c(s_0,t_n^-) + c(t_n,t_n^+)) \geq M'$.
   \end{itemize}
   Hence, we have
   $M \PNtran{t_1} M_1 \PNtran{t_2} \ldots \PNtran{t_n} M_n$.
   As $M_n \geq M'$, $M'$ is reachable from $M$.
 \end{proof}

We have the following result:

\begin{corollary}
 The lower bound for the model-checking problem complexity of \atlr is EXPSPACE.
\end{corollary}

\section{Feasible cases}
\label{sec:feasible}

In the previous section, we have seen that the model-checking problem for \atlr is EXPSPACE-hard.
There are, however, several tractable special cases of the model-checking problem. Here we consider two of them:
model-checking \atlr with a single resource, and model-checking \RBATL (\atlr with only consumption of resources).

\subsection{Model-checking \atlr with a single resource}

For the case when $|Res| = 1$, the problem
whether $M, s \models \phi_0$ is decidable in PSPACE. 

\begin{theorem}
The upper bound for the model-checking problem complexity of \atlr with a single resource is PSPACE.
\end{theorem}
\begin{proof}
All the cases in Algorithm~\ref{alg:atlr-label} apart from $\AU{A^b}{\phi}{\psi}$ and $\AG{A^b}{\phi}$
can be computed in time polynomial in $|M|$ and $|\phi|$. The cases for $\AU{A^b}{\phi}{\psi}$ and $\AG{A^b}{\phi}$
are more computationally expensive. They involve calling the \textsc{until-strategy} and the \textsc{box-strategy} procedures,
respectively, for every state in $S$. The procedures explore the model in a depth-first manner, one path at a time.
Their space requirement corresponds to the maximal length of such a path. Note that unlike depth-first search,
\textsc{until-strategy} and \textsc{box-strategy} in the general case (multiple resources) do not terminate when 
they encounter a loop, that is a path containing two nodes with the same state: $\ldots, n_1,\ldots, n_2$ where $s(n_1)
= s(n_2)$, since in the general case $e(n_1)$ and $e(n_2)$ may be incomparable. However, for a single resource,
it will always be the case that either $e(n_1)=e(n_2)$, or $e(n_1) < e(n_2)$, 
or $e(n_1) > e(n_2)$.  Inspection 
of \textsc{until-strategy} and \textsc{box-strategy} shows that they will return in all of these cases.
Hence, we never need to keep a stack of more than $|S|$ nodes, which requires polynomial space. 
\end{proof}
The result above can be generalised to the case when $|Res|>1$, but the formula
$\phi_0$ is of a special form, where at most one resource is non-$\infty$ in
each bound. To be precise, $\phi_0$ is such that in each
resource bound $b$ occurring in $\phi_0$, for at most one resource 
$i$, $b_i \not = \emptyset$. 

\subsection{Model-checking \RBATL}

In this section, we briefly revisit the problem of model-checking \RBATL (the
logic where only consumption of resources is considered).
The syntax of \RBATL is the same
as the syntax of \atlr, and the models are the class of RB-CGS with no
production of resource (all action costs are non-negative). 
We will refer to such models as RB-CGS${}^-$.
A symbolic model-checking algorithm
for that logic was introduced in \cite{Alechina//:10a} (without infinite resource bounds). Here we re-state the algorithm and discuss upper and lower bounds on the complexity of \RBATL model-checking.
  
The algorithm uses an abbreviation $split(b)$ that takes a resource 
bound $b$ and returns the set of all pairs $(d,d') \in \nat_{\infty} \times 
\nat_{\infty}$ such that:
\begin{enumerate}
\item $d + d' = b$,
\item $d_i = d'_i = \infty$ for all $i \in \{ 1, \ldots, r\}$ such that $b_i = \infty$, and
\item $d$ has at least one non-0 value.
\end{enumerate} 
We assume that $split(b)$ is partially ordered in increasing order of the 
second component $d'$ (so that if $d'_1 < d'_2$, then $(d_1,d'_1)$ precedes 
$(d_2,d'_2)$).

The algorithm is similar to the symbolic model-checking algorithm for \atl\ given in
\cite{Alur//:02a}. The main differences from the algorithm for \atl\
is the addition of costs of actions, and, instead of
working with a straightforward set of subformulas $Sub(\phi_0)$ of a given formula $\phi_0$,
we work with an extended set of subformulas $Sub^+(\phi_0)$. $Sub^+(\phi_0)$ includes
$Sub(\phi_0)$, and in addition:
\begin{itemize}
\item if $\atlbox{A^b}\phi \in Sub(\phi_0)$, then
$\atlbox{A^{d'}}\phi \in Sub^+(\phi_0)$ for all $d'$ such that
$(d,d') \in split(b)$;
\item if $\atlu{A^b}{\phi}{\psi} \in Sub(\phi_0)$, then
$\atlu{A^{d'}}{\phi}{\psi} \in Sub^+(\phi_0)$ for all $d'$ such that
$(d,d') \in split(b)$.
\end{itemize}
We assume that $Sub^+(\phi_0)$ is ordered in the increasing order of
complexity and of resource bounds (so \eg for $b \leq b'$, $\atlbox{A^b}\psi$
precedes $\atlbox{A^{b'}}\psi$).

\begin{theorem}
Given an RB-CGS${}^-$ $M=(Agt, Res, S, \Pi, \pi, Act, d, c, \delta)$ and 
an \RBATL
formula $\phi_0$, there is an
algorithm which returns the set of states $[\phi_0]_M$ satisfying
$\phi_0$: $[\phi_0]_M = \{s\ |\ M,s \models \phi_0\}$, which runs in time
$O(|\phi_0|^r \times m)$ where $r$ is $|Res|$ and $m$ is the number of transitions in $M$,
assuming that numbers in bounds are encoded in unary.
\end{theorem}
\begin{proof}
  Let $\projinf{\zerob}{b}$ be a vector where the $i$th component is
$\infty$ if the $i$th component of $b$ is $\infty$, and $0$ otherwise.  
Let $Pre$ as before be a
  function which given a coalition $A$, a set $\rho \subseteq S$ and a
  bound $b$ returns a set of states $s$ in which $A$ has a move
  $\sigma_A$ with cost $cost(s,\sigma_A) \leq b$ such that
  $out(s,\sigma_A) \subseteq \rho$. 
  Consider Algorithm \ref{alg:rbatl-label}.
\begin{algorithm}
\caption{Model-checking \RBATL}
\label{alg:rbatl-label}
\begin{algorithmic}
\Function{rb-atl-label}{$M, \phi_0$}
\For{$\phi' \in Sub^+(\phi)$}
\Case{$\phi' = p,\ \neg \phi,\ \phi \wedge \psi$}
\State standard, see \cite{Alur//:02a}
\EndCase
\Case{$\phi' = \AO{A^b}{\phi}$} \ \ \ \ $[\phi']_M \gets Pre(A, [\phi]_M, b)$
\EndCase
\Case{$\phi' = \AU{A^b}{\phi}{\psi}$} where $b$ is such that
for all $i$, $b_i \in \{0, \infty\}$:
\State $\quad \rho \gets [false]_M;  \tau \gets [\psi]_M$;
\State $\quad \mathbf{while}\  \tau \not \subseteq \rho \ \mathbf{do}$
\State $\quad \quad \rho \gets \rho \cup \tau; \tau \gets Pre(A,\rho,b) \cap [\phi]_M$ 
\State $\quad \mathbf{od}$ 
\State $\quad [\phi']_M \gets \rho$
\EndCase
 \Case{$\phi' = \AU{A^b}{\phi}{\psi}$} where $b$ is such that for
 some $i$, $b_i \not \in \{0, \infty\}$:
 \State $\quad \rho \gets [false]_M;  \tau \gets [false]_M$
 \State $\quad \mathbf{foreach}\  d' \in \{d' \mid (d,d') \in split(b)\} \ \mathbf{do}$
 \State $\quad \quad \tau \gets Pre(A,[\atlu{A^{d'}}{\phi}{\psi}]_M,d) \cap [\phi]_M$
 \State $\quad \quad \mathbf{while}\  \tau \not \subseteq \rho \ \mathbf{do}$
 \State $\quad \quad \quad \rho \gets \rho \cup \tau; \tau \gets Pre(A,\rho,\projinf{\zerob}{b}) \cap [\phi]_M$
\State $\quad \quad \mathbf{od}$
\State $\quad \mathbf{od}$
\State $\quad [\phi']_M \gets \rho$
\EndCase
\Case{$\phi' = \AG{A^b}{\phi}$} where  $b$ is such that for all $i$,
$b_i \in \{0, \infty\}$:
\State $\quad \rho \gets [true]_M;  \tau \gets  [\phi]_M$ 
\State $\quad \mathbf{while}\  \rho \not \subseteq \tau\ \mathbf{do}$
\State $\quad \quad \rho \gets \tau; \tau \gets Pre(A,\rho,b) \cap [\phi]_M$
\State $\quad \mathbf{od}$   
\State $\quad [\phi']_M \gets \rho$
\EndCase
\Case{$\phi' = \AG{A^b}{\phi}$} where $b$ is such that for some $i$, $b_i \not \in \{0, \infty\}$:
\State $\quad \rho \gets [false]_M;  \tau \gets [false]_M$ 
\State $\quad \mathbf{foreach}\ d' \in \{d' \mid (d,d') \in split(b)\} \ \mathbf{do}$
\State $\quad \quad \tau \gets Pre(A,[\atlbox{A^{d'}}\phi]_M,d) \cap [\phi]_M$
\State $\quad \quad \mathbf{while}\  \tau \not \subseteq \rho \ \mathbf{do}$
\State $\quad \quad \quad \rho \gets \rho \cup \tau; \tau \gets Pre(A,\rho,\projinf{\zerob}{b})
\cap [\phi]_M$
\State $\quad \quad \mathbf{od}$
\State $\quad \mathbf{od}$
\State $\quad [\phi']_M \gets \rho$
\EndCase
\EndFor
\State $\mathbf{return\ } [\phi_0]_M$
\EndFunction
\end{algorithmic}
\end{algorithm}
Note that $|split(b)|$ is $O(\beta^r)$, where $\beta$ is the largest component occurring in $b$. If $\phi_0$ contains operators with
bounds containing components other than $0$ and $\infty$, $|Sub^+(\phi_0)|$ is $O(|\phi_0| \times |\beta|^r)$, or $O(|\phi_0| \times |\phi_0|^r)$
provided that vector components are encoded in unary.
This moves the complexity from $O(|\phi_0|\times m)$ as in \cite{Alur//:02a} to 
$O(|\phi_0|^r\times m)$, where $m$ is the number of transitions in $M$. See \cite{Alur//:02a} for the argument.
\end{proof}


\section{Comparison with \ral}
\label{sec:ral}

In this section, we compare \atlr with the logics introduced in 
\cite{Bulling/Farwer:10a}, in particular with the logic \pfral. 
In \cite{Bulling/Farwer:10a}, it is shown that the model-checking
problem for \pfral with infinite semantics is undecidable. The decidability
of the model-checking problem for \pfral with finite semantics is
stated in \cite{Bulling/Farwer:10a} as an open problem. Here we show that 
model-checking for \pfral with finite semantics is decidable.


\subsection{The logic \pfral}

The logical language \pfral is a proponent-restricted and resource-flat version
of \ral without the release operator (for a complete description of \ral
and its variants, we refer the reader to \cite{Bulling/Farwer:10a} and its technical
report version \cite{Bulling/Farwer:10b}; in fact the name \pfral comes from
\cite{Bulling/Farwer:10b}).



The syntax of \pfral is defined using \emph{endowment functions}
(or just endowments) rather than resource bounds.
An endowment is a function $\eta : \Agt \times \Res \to
\nat \cup \{\infty\}$. 
We will sometimes write $\eta_a(r)$ instead of $\eta(a,r)$. Let $\En$ denote the
set of all possible endowments.

Formulas of \pfral are defined as follows:
\[
\phi,\psi ::= p \mid \neg \phi \mid \phi \land \psi \mid 
            \ralAO A \eta \phi \mid
            \ralAG A \eta \phi \mid
            \ralAU A \eta \phi \psi 
\]
where $p \in \Pi$, $A \subseteq \Agt$, $A \not= \emptyset$ and  $\eta \in \En$.


Formulas of \pfral are interpreted on resource-bounded models (RBM) which are
CGS structures extended with resources except that transitions are in general not total,
i.e., at a state, an agent is not required to have any available actions. This
means that there may be a state in an RBM model which does not have
any successor. An RBM is defined as follows:

\begin{definition}
An RBM is a tuple $M = (\Agt, Q, \Pi, \pi, \Act, d, o, \Res, t)$ where $\Agt$, $\Act$,
$Q$, $\Pi$, $\Res$, and $o$ are defined as $\Agt$, $\Act$ except that $\idle$ is not required to be in $\Act$, $S$, $\Pi$, $\Res$, and
$\delta$, respectively, in Definition~\ref{def:rbcgs} and:
\begin{itemize}
 \item $\pi : Q \to \wp(\Pi)$ specifies propositional valuation;
 \item $d  : \Agt \times Q \to \wp(\Act)$ specifies available actions;
 \item $t : \Act \times \Res \to \integer$ for an action $\alpha \in \Act$
       and a resource $r \in \Res$ specifies the consumption of $r$ by $\alpha$ if
       $t(\alpha,r) \leq 0$ or the production of $r$ by $\alpha$ if $t(\alpha,r) > 0$.
       Let $\cons(\alpha,r) = -\min\{0,t(\alpha,r)\}$ and $\produce(\alpha,r) =
       \max\{0,t(\alpha,r)\}$.
\end{itemize}
\end{definition}

Resource availability is modelled by resource-quantity mappings (rqm) $\rho
: \Res \to \integer \cup \{\infty\}$.

Given a RBM $M$, $Q^{\leq \omega} = Q^\omega \cup Q^+$ denotes the set of all
finite and infinite sequences over $Q$. A sequence $\lambda \in Q^{\leq \omega}$
is a \emph{path} in $M$ iff there exist transitions in $M$ between adjacent states
in $\lambda$. A finite or infinite sequence $\lambda =
(q_0,\eta^0),(q_1,\eta^1),\ldots$ over $Q \times \En$ is a \emph{resource-extended
path} (r-path) in $M$ iff $q_0,q_1,\ldots$ is a path in $M$.

Given a coalition $A$, an endowment $\eta$ and an rqm $\rho$,
an $(A,\eta)$-share for $\rho$ is a function $\sh : A \times \Res \to \nat$ where:
\begin{itemize}
 \item $\forall r \in \Res: \rho(r) > 0 \Rightarrow \sum_{a \in A} \sh(a,r) = \rho(r)$;
 \item $\forall a \in A, r \in \Res: \eta_a(r) \geq \sh(a,r)$.
\end{itemize}
Let $\Share(A,\eta,\rho)$ denote the set of all possible $(A,\eta)$-shares for
$\rho$. It is straightforward that $\Share(A,\eta,\rho) = \emptyset$ if
$\sum_{a\in A}\eta_a(r) < \rho(r)$, i.e., resource endowment for agents in $A$
is not enough to create a share.

Given an endowment $\eta$ and a strategy $F_A$ for a coalition $A$, a maximal
r-path $\lambda = (q_0,\eta^0),(q_1,\eta^1),\ldots$ of $M$ is an
$(\eta,F_A)$-path starting from a state $q_0$ iff:
\begin{itemize}
 \item $\eta^0 = \eta$;
 \item $\forall a \in A, r \in \Res, i \geq 0, i < |\lambda|:  \eta^i_a(r) \geq 0$;
 \item $\forall i \geq 0, i < |\lambda|-1: \exists \sigma \in D(q_i)$ such that:
 \begin{itemize}
  \item $F_A(q_0\ldots q_i) = \sigma_A$;
  \item $o(q_i,\sigma) = q_{i+1}$;
  \item $\exists \sh_i \in \Share(A,\eta,\rho): \forall a \in A, r \in \Res: 
        \eta^{i+1}_a(r) = \eta^i_a(r) + \produce(\sigma_a,r) - \sh_i(a,r)$
        where $\rho$ is such that $\rho(r) =  \sum_{a\in A} -\cons(\sigma_a,r)$.
 \end{itemize}
\end{itemize}
Notice that as defined in \cite{Bulling/Farwer:10a}, a path is maximal if it can be extended with sufficient available resources, then it must be extended.
Then, $out(q_0,\eta,F_A)$ denotes the set of all $(\eta,F_A)$-paths starting
from a state $q_0$. As shown by \cite{Bulling/Farwer:10a}, $out(q_0,\eta,F_A)$
is never empty. In the worst case, $out(q_0,\eta,F_A)$ contains a single r-path
$(q_0,\eta)$.

Given an RBM $M$ and a state $q$, the truth of \pfral formulas is defined
inductively as follows (we omit the propositional cases):
\begin{itemize}
 \item $M,q \models_{ral} \ralAO A \eta \phi$ iff 
       $\exists F_A: \forall \lambda \in out(q,\eta,F_A):
                     |\lambda| \geq 2 \land
                     M, \lambda[1] \models_{ral} \phi$;
 \item $M,q \models_{ral} \ralAG A \eta \phi$ iff 
       $\exists F_A: \forall \lambda \in out(q,\eta,F_A):
                     |\lambda| = \infty \land
                     \forall i \geq 0: M, \lambda[i] \models_{ral} \phi$;
 \item $M,q \models_{ral} \ralAU A \eta \phi \psi$ iff 
       $\exists F_A: \forall \lambda \in out(q,\eta,F_A):
                     \exists i \geq 0, i < |\lambda|:
                     M, \lambda[i] \models_{ral} \psi \land
                     \forall j \geq 0, j < i: M, \lambda[j] \models_{ral} \phi$.
\end{itemize}

The definition above gives \emph{finite semantics} of \pfral. 
\emph{Infinite semantics} is obtained if the condition ``for all
$\lambda \in out(q,\eta,F_A)$'' above is replaced with ``for all infinite 
$\lambda \in out(q, \eta, F_A)$''.

\begin{theorem}\cite{Bulling/Farwer:10a,Bulling/Farwer:10b}
The model-checking problem for \pfral with infinite semantics is undecidable.
\end{theorem}

The problem whether model-checking for \pfral with finite semantics is 
decidable is left open in \cite{Bulling/Farwer:10a}. Below we show that it is
in fact decidable by adapting the model-checking algorithm for \atlr. Before
we do this, we investigate the differences between \pfral and \atlr 
in more detail. In particular we consider whether we can obtain
a logic equivalent to \pfral by simply removing
the restriction that agents always have at least the $idle$ action available
from the semantics of \atlr.

\subsection{The logic \atlrnt}

As models for \pfral are not total in general, we facilitate a comparison 
with \atlr by
introducing a variant \atlrnt of \atlr where we remove the requirement of total
transitions in Definition \ref{def:rbcgs}. In other words, \atlrnt has the same
syntax as \atlr yet a broader class, namely RB-CGS-nt, of models which do not
need to be total. In particular, in Definition \ref{def:rbcgs}, $\Act$ does not
need to include $\idle$ and $d: S \times \Agt \to \wp(\Act)$ may be mapped to an
empty set or to a set not containing $\idle$.

Obviously, any RB-CGS model is an RB-CGS-nt but not vice versa. Since RB-CGS-nt
models are not total in general, at a state $s$, the set $D_A(s)$ of possible
joint actions by a coalition $A$ and the set of possible outcomes of a joint
action $\sigma_A \in D_A(s)$ may be empty.

Given a RB-CGS-nt model $M$, a strategy $F_A$ for a coalition $A \subseteq
\Agt$, a finite computation $\lambda \in S^+$ is consistent with $F_A$ iff for
all $i \in \{0,\ldots, |\lambda|-2\}$: $\lambda[i+1] \in
out(\lambda[i],F(\lambda[0,i]))$ and $D_{\Agt}(\lambda[|\lambda|-1]) =
\emptyset$, i.e., there is a deadlock at the last state of $\lambda$. We denote
by $out_f(s,F_A)$ the set of all consistent finite computations of $F_A$
starting from $s$. Then, the set of all consistent finite and infinite
computations of $F_A$ from $s$ is defined as:
\[
out_{nt}(s,F_A) = out(s,F_A) \cup out_f(s,F_A)
\]

Under a resource bound $b \in B$, a computation $\lambda \in out_{nt}(s,F_A)$ can be
only carried out until an index $i_{\max} \in \nat_\infty$ (see Figure
\ref{fig:b-consistent}) iff:
\[
\sum_{j=0}^{i} cost(\lambda[j], F_A(\lambda[0,j])) \leq b \text{ for all } i < i_{\max}
\]
and
\[
\sum_{j=0}^{i_{\max}} cost(\lambda[j], F_A(\lambda[0,j])) \not\leq b \text{ if } i_{\max} \not= \infty
\]

\begin{figure}[h]
\centering
\def\svgwidth{0.9\textwidth}
{\tiny
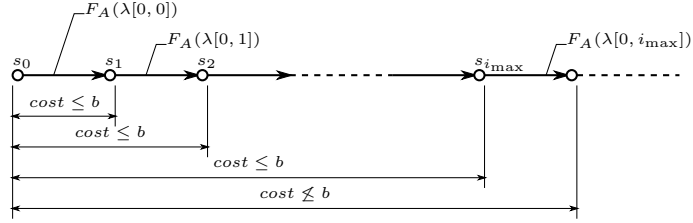
}
\caption{$\lambda$ is restricted by $b$.}
\label{fig:b-consistent}
\end{figure}

Let us denote $\lambda(b) = \lambda[0,i_{\max}]$ and we call $\lambda(b)$ 
maximal with respect to $b$. Then, the set of all $b$-consistent (finite or
infinite) computations of $F_A$ starting from state $s$ is defined as follows:
\[
out_{nt}(s,F_A,b) = \{ \lambda(b) \mid \lambda \in out_{nt}(s,F_A) \}
\]
Note that this definition implies that the cost of every prefix of a
$b$-consistent computation is below $b$ and $out_{nt}(s,F_A,b)$ may contain finite
computations. Furthermore, $out_{nt}(s,F_A,b)$ is always non-empty, as in the worst
case, it contains a single computation $s$. 

The semantics of \atlrnt formulas is defined as follows (the atomic case and
Boolean connectives are defined in the standard way):
\begin{itemize}
\item $M, s \models_{nt} \AO{A^b}{\phi}$ iff $\exists$ strategy $F_A'$ such that for
      all $\lambda \in out_{nt}(s, F_A',b)$: $|\lambda| \geq 2$ and $M, \lambda[1]
      \models \phi$;

\item $M, s \models_{nt} \AG{A^b}{\phi}$ iff $\exists$ strategy $F_A'$ such that for
      all $\lambda \in out_{nt}(s, F_A',b)$ and $i \geq 0$: $|\lambda|=\infty$ and $M,
      \lambda[i] \models \phi$; and

\item $M, s \models_{nt} \AU{A^b}{\phi}{\psi}$ iff $\exists$ strategy $F_A'$ such
      that for all $\lambda \in out_{nt}(s, F_A',b)$, $\exists i \geq 0$: $i <
      |\lambda|$, $M, \lambda[i] \models \psi$ and $M, \lambda[j] \models \phi$
      for all $j \in \{0,\ldots,i-1\}$.

\end{itemize}
If the condition ``for all $\lambda \in out_{nt}(s, F_A',b)$'' is replaced with
``for all infinite $\lambda \in out_{nt}(s, F_A',b)$'' in the truth definition
of \atlrnt, we obtain \atlrnt with infinite semantics. Note that in a RB-CGS
model $M$, if $F_A$ is a $b$-strategy for a coalition $A$, we have that
$out(s,F_A) = out(s,F_A,b) = out_{nt}(s,F_A) = out_{nt}(s,F_A,b)$.
We have the following result:

\begin{lemma}
\label{lemma:atlr-ag-atlrnt}
Given a RB-CGS model $M$,
$M, s \models {\phi'}$ iff $M, s \models_{nt} {\phi'}$ under finite semantics.
\end{lemma}

\begin{proof}
$(\Rightarrow)$ is obvious. For $(\Leftarrow)$, the proof is by induction on the
structure of $\phi'$. 

If $\phi' = \AG{A^b}{\phi}$, we have that $out_{nt}(s,F_A') = out_{nt}(s, F_A',b)$ because  $|\lambda| = \infty$ for all $\lambda
\in out_{nt}(s, F_A',b)$;
thus, $F_A'$ is a $b$-strategy.

If $\phi' = \AO{A^b}{\phi}$, let us consider the following strategy for $A$:
\[
F_A(\lambda) = 
\begin{cases}
F_A'(\lambda) & \text{ if } \exists \lambda' \in S^+ \cup S^\omega: \lambda\lambda' \in out(s,F_A',b) \\
idle          & \text{ otherwise.}
\end{cases}
\]
It is straightforward that $F_A$ is a $b$-strategy to satisfy $M, s \models
\AO{A^b}{\phi}$ at $s$.

If $\phi' = \AU{A^b}{\phi}{\psi}$, the proof is similar to the above case,
hence it is omitted here.
\end{proof}

The above result shows that over the class of RB-CGS models, \atlr and \atlrnt
with finite semantics are equivalent. Furthermore, we have the following result:

\begin{theorem}
The model-checking problem for \atlrnt with finite semantics is decidable.
\end{theorem} 
\begin{proof}
The model-checking algorithm for \atlr can be easily adapted
to a model-checking algorithm for \atlrnt. The only change required is in 
the function $Pre(A,\rho,b)$ which becomes
$
Pre(A,\rho,b) = \{ s\in S \mid \exists \sigma_A \in D_A(s): cost(s, \sigma_A) 
\leq b \land \emptyset \not= out(s,\sigma_A) \subseteq \rho \}.
$
Here, we additionally require that $out(s,\sigma_A) \not= \emptyset$.
\end{proof}

\subsection{Comparing \pfral and \atlrnt}
\sloppy

At the semantical level, every RBM $M = (\Agt, Q, \Pi, \pi, \Act, d,
o, \Res, t)$ can be converted straightforwardly into an RB-CGS $M' =
(\Agt, \Res, Q, \Pi, \pi', \Act, d, c,$ $ \delta)$ where:
\begin{itemize}
 \item $\pi'(p)= \{ q \in Q \mid p \in \pi(q) \}$ for all $p \in \Pi$;
 \item $c(q, a, \alpha) = (-t(\alpha,r))_{r \in \Res}$ for all $q \in Q, a \in \Agt, \alpha \in \Act$; and
 \item $\delta = o$.
\end{itemize}

At the syntactical level, \pfral and \atlrnt are rather different. While \pfral
enables specifying the ability of a coalition under a resource endowment,
\atlrnt allows one to specify the ability of a coalition within a resource
bound. Let us consider an example, as depicted in Figure~\ref{fig:pfral-atlnt},
\begin{figure}
\centering
\def\svgwidth{0.5\textwidth}
{
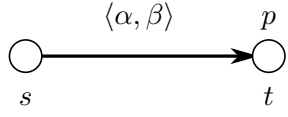}
\caption{Comparing resource endowments and bounds.}\label{fig:pfral-atlnt}
\end{figure}
in order to clarify the difference between endowments and bounds. In this
example, our model has two agents $a$ and $b$ and one resource. From state $s$,
agents $a$ and $b$ can only perform $\alpha$ and $\beta$, respectively, which
cost $-c$ and $c$ (for some $c > 0$), respectively. As their joint action is
cost-free, we have that $s \models_{nt} \AO{\{a,b\}^0} p$. However, given an
empty endowment $\eta_0 = \{a \mapsto 0, b \mapsto 0\}$, there is no possible
share from this endowment to cover the cost $c$ of action $\beta$; i.e., $s
\not\models_{ral} \ralAO {\{a,b\}} {\eta_0} p$. The reason is that under
$\eta_0$ $(s,\eta_0)$ is the only from $s$ which is shorter that the computation
$s,t$ under $0$. In general, we have the following result: 
\begin{lemma}
\label{lemma:compare-computations}
Given a RBM model $M$, for any state $q_0$,
strategy $F_A$, endowment $\eta^0$ and bound $b = (\sum_{a \in A}\eta_a(r))_{r
\in \Res}$, then if $(q_0,\eta^0),(q_1,\eta^1),\ldots,(q_k,\eta^k)$ is the prefix
of some computation in $out(q_0,\eta^0,F_A)$, then $q_0,q_1,\ldots,q_k$ is also
the prefix of some computation in $out_{nt}(q_0,F_A,b)$.
\end{lemma}
\begin{proof}
The proof is done by induction on $k$; additionally, we also show that
$(\sum_{a \in A}\eta^k_a(r))_{r\in \Res} = b - \sum_{j=0}^{k-1}
cost(q_j,F_A(q_1\ldots q_j))$.
\begin{description}
 \item [Base case $k=0$:] The proof is trivial.
 \item [Induction step:] Assume that
       $(q_0,\eta^0)(q_1,\eta^1)\ldots(q_{k+1},\eta^{k+1})$ is the prefix of
       some computation in $out(q_0,\eta^0,F_A)$. Then, so is
       $(q_0,\eta^0)(q_1,\eta^1)\ldots(q_k,\eta^{k})$. By induction hypothesis,
       we have that $q_0\ldots q_k$ is the prefix of some computation in
       $out_{nt}(q_0,F_A,b)$ and $(\sum_{a \in A}\eta^{k}_a(r))_{r\in \Res} = b -
       \sum_{j=0}^{k-1} cost(q_j,F_A(q_1\ldots q_j))$.
       
       As $(q_0,\eta^0)(q_1,\eta^1)\ldots(q_{k+1},\eta^{k+1})$ is a prefix,
       $\Share(A,\eta^k,(\sum_{a \in A}-\cons(F_A(q_0\ldots q_k)_a,r))_{r \in
       \Res}) \not = \emptyset$, i.e., $\sum_{a \in A}\eta^k_a(r) \geq \sum_{a
       \in A}-\cons(F_A(q_0\ldots q_k)_a,r)$ for all $r \in \Res$; hence
       $\sum_{a \in A}\eta^{k+1}_a(r) \geq \sum_{a \in A}\produce(F_A(q_0\ldots
       q_k)_a,r) \geq 0$.
       
       We also have $(\sum_{a \in A}\eta^{k+1}_a(r))_{r\in \Res} = (\sum_{a \in
       A}(\eta^{k}_a(r) + \produce(F_A(q_0\ldots q_k)_a,r) - \sh_k(a,r)))_{r\in
       \Res} = (\sum_{a \in A}(\eta^{k}_a(r) + \produce(F_A(q_0\ldots q_k)_a,r)
       + \cons(F_A(q_0\ldots q_k)_a,r)))_{r\in \Res} = b - \sum_{j=0}^{k}
       cost(q_j,F_A(q_1\ldots q_j))$. As $\sum_{a \in A}\eta^{k+1}_a(r) \geq 0$
       for all $r \in \Res$, $b - \sum_{j=0}^{k} cost(q_j,F_A(q_1\ldots q_j))
       \geq 0$, i.e., $\sum_{j=0}^{k} cost(q_j,F_A(q_1\ldots q_j)) \leq b$,
       hence $q_0\ldots q_{k+1}$ is also a prefix of some computation in
       $out_{nt}(q_0,F_A,b)$.
\end{description}
\end{proof}

As suggested by the function $\eta^b$ which translates resource bounds into
endowments (introduced in \cite{Bulling/Farwer:10b} by Bulling and Farwer to
relate their framework to RBCL \cite{Alechina//:09b}), \pfral formulas can also
be converted into \atlrnt formulas by a translation function $\tr$ which makes
use of the inverse of $\eta^b$ and is defined inductively as follows
(propositional cases are omitted):
\begin{itemize}
 \item $\tr(\ralAO A \eta \phi) = \AO{A^{(\sum_{a \in A}\eta_a(r))_{r \in \Res}}} {\tr(\phi)}$;
 \item $\tr(\ralAG A \eta \phi) = \AG{A^{(\sum_{a \in A}\eta_a(r))_{r \in \Res}}} {\tr(\phi)}$; and
 \item $\tr(\ralAU A \eta \phi \psi) = \AU{A^{(\sum_{a \in A}\eta_a(r))_{r \in \Res}}} {\tr(\phi)} {\tr(\psi)}$.
\end{itemize}
Here, resource bounds are sums of individual endowments for each resource. The
example in Figure~\ref{fig:pfral-atlnt} and
Lemma~\ref{lemma:compare-computations} show that satisfiability is not preserved
by the translation function $\tr$. In order to obtain preservation of
satisfiability, it is necessary
to relax the requirement in the definition of computations in RBM models. In particular,
the last condition is relaxed as follows:
\begin{itemize}
\item $\exists \sh_i \in \Share(A,\eta,\rho): \forall a \in A, r \in \Res:
      \eta^{i+1}_a(r) = \eta^i_a(r) + \produce(\sigma_a,r) - \sh_i(a,r)$ where
      $\sigma = F_A(q_0\ldots q_i)$ and
      $\rho(r) = \sum_{a\in A} (-\cons(\sigma_a,r)-\produce(\sigma_a,r))$.
\end{itemize}
Comparing the the original condition, the production of resource in a step is
also considered to cover for the consumption in the same step by adding it into
the share function. Let us call RBM models with this relaxed condition 
relaxed RBM models. We have the following result:
\begin{lemma}
\label{lemma:compare-relaxed-computations}
Given a relaxed RBM model $M$, for any state $q_0$, strategy $F_A$, endowment
$\eta^0$ and bound $b = (\sum_{a \in A}\eta_a(r))_{r \in \Res}$, then:
\begin{itemize}
\item if $(q_0,\eta^0),(q_1,\eta^1),\ldots \in out(q_0,\eta^0,F_A)$, then
      $q_0,q_1,\ldots\in out(q_0,F_A,b)$;
\item conversely, if $q_0,q_1,\ldots\in out_{nt}(q_0,F_A,b)$; then $\exists
      \eta^1,\eta^2\ldots$ such that $(q_0,\eta^0),(q_1,\eta^1),\ldots \in
      out(q_0,\eta^0,F_A)$.
\end{itemize}
\end{lemma}
\begin{proof}
Both directions are repetition of the proof of
Lemma~\ref{lemma:compare-computations}, hence they are omitted here.
\end{proof}

Let \pfralrelax be \pfral interpreted over relaxed RBM models. We have the following result:
\begin{lemma}
\label{lemma:preserve-sat}
Given a relaxed RBM model $M$, $M,s \models_{\text{\pfralrelax}} \phi'$ iff $M',s \models_{nt} \tr(\phi')$.
\end{lemma}

\begin{proof}
Let us prove the direction from left to right. The other direction is similar.
The proof is done by induction on the structure of $\phi'$. The base case is
trivial, hence omitted here.

In the induction step, the cases of propositional connectives are trivial, hence
they are also omitted. Let us consider the following three cases.

\begin{description}
\item [$\phi' = \ralAO A \eta \phi$: ] Let $b = (\sum_{a\in A}
      \eta_a(r))_{r\in \Res}$ and $F_A$ be the strategy to satisfy $\phi'$ at
      $s$. For every $q_0q_1\ldots \in out_{nt}(s,F_A,b)$ where $s = q_0$, by
      Lemma~\ref{lemma:compare-relaxed-computations}, there are $\eta^1,\eta^2$
      such that $(q_0,\eta^0)(q_1,\eta^1)\ldots \in out(s,\eta,F_A)$. As $M, s
      \models_{ral} \ralAO A \eta \phi$, we have that $M,q_1 \models_{ral}
      \phi$. By induction hypothesis, $M', q_1\models_{nt} \tr(\phi)$. Hence,
      $M',s \models \AO{A^{b}}\tr(\phi)$

\item [$\phi' = \ralAU A \eta \phi \psi$: ] Let $b = (\sum_{a\in A}
      \eta_a(r))_{r\in \Res}$ and $F_A$ be the strategy to satisfy $\phi'$ at
      $s$. For every $q_0q_1\ldots \in out_{nt}(s,F_A,b)$ where $s = q_0$, by
      Lemma~\ref{lemma:compare-relaxed-computations}, there are $\eta^1,\eta^2$
      such that $(q_0,\eta^0)(q_1,\eta^1)\ldots \in out(s,\eta,F_A)$. As $M, s
      \models_{ral} \ralAU A \eta \phi \psi$, we have that $\exists i \geq 0$ such
      that $M,q_j \models_{ral} \phi$ for all $j < i$ and $M,q_i \models_{ral}
      \psi$. By induction hypothesis, $M',q_j \models_{nt} \tr(\phi)$ for all $j
      < i$ and $M,q_i \models_{ral} \psi$. Hence, $M',s \models
      \AU{A^{b}}{\tr(\phi)}{\tr(\psi)}$.

\item [$\phi' = \ralAG A \eta \phi$: ] Let $b = (\sum_{a\in A}
      \eta_a(r))_{r\in \Res}$ and $F_A$ be the strategy to satisfy $\phi'$ at
      $s$. For every $q_0q_1\ldots \in out_{nt}(s,F_A,b)$ where $s = q_0$, by
      Lemma~\ref{lemma:compare-relaxed-computations}, there are $\eta^1,\eta^2$
      such that $(q_0,\eta^0)(q_1,\eta^1)\ldots \in out(s,\eta,F_A)$. As $M, s
      \models_{ral} \ralAU A \eta \phi \psi$, we have that
      $(q_0,\eta^0)(q_1,\eta^1)\ldots$ is infinite and $M,q_j \models_{ral}
      \phi$ for all $j \geq 0$. By induction hypothesis, $M',q_j \models_{nt}
      \tr(\phi)$ for all $j \geq 0$. Hence, $M',s \models
      \AG{A^{b}}{\tr(\phi)}$.
\end{description}

\end{proof}
The above lemma shows that over the class of relaxed RBM models, \atlrnt and
\pfralrelax with finite semantics are equivalent. Similar to the above result,
it is also straightforward that \atlrnt with infinite semantics is equivalent to
\pfralrelax with infinite semantics:
\begin{lemma}
Given a relaxed RBM model $M$, under the infinite semantics, 
$M,s \models_{\text{\pfralrelax}} \phi'$ iff $M',s \models_{nt}
\tr(\phi')$.
\end{lemma}
\begin{proof}
The proof is the same as the proof of Lemma~\ref{lemma:preserve-sat} except we
only consider infinite computations.
\end{proof}
Note that the proof for the undecidability
of \pfral in~\cite{Bulling/Farwer:10a} with infinite semantics can be applied
for \pfralrelax with infinite semantics. Hence, we have the following result:
\begin{lemma}
Model-checking \pfralrelax with infinite semantics is undecidable.
\end{lemma}
Then, we have the following consequences:
\begin{corollary}
Model-checking \atlrnt with infinite semantics is undecidable.
\end{corollary}
Since model-checking \atlrnt with finite semantics is decidable, we have:
\begin{lemma}
Model-checking \pfralrelax with finite semantics is decidable.
\end{lemma}
Furthermore, the same result can also be established for \pfral:
\begin{theorem}
Model-checking \pfral with finite semantics is decidable.
\end{theorem}
\begin{proof}
We adapt further the model-checking algorithm for \atlrnt where the line $ActA 
\leftarrow \{ \sigma \in D_A(s(n)) \mid cost(\sigma) \leq e(n) \}$ in
Algorithms~\ref{alg:until-strategy-multiple-resources} and
\ref{alg:box-strategy-multiple-resources} is replaced by $ActA \leftarrow \{\sigma 
\in D_A(s(n)) \mid (\sum_{i \in A}\cons(\sigma_i,r))_{r \in \Res} \leq e(n) \}$.
\end{proof}
Figure~\ref{table:mc-results} summarises the above decidability and undecidability results for the model-checking problems for \atlr,
\atlrnt, \pfralrelax and \pfral where D stands for decidable and U for
undecidable. Note that \atlr is decidable in both semantics due to the fact that
both semantics are indistinguishable thanks to $\idle$.
\begin{figure}[h]
\centering
\begin{tabular}{|l|c|c|c|c|}
\hline
Semantics & \atlr & \atlrnt & \pfralrelax & \pfral \\
\hline
\hline
Finite     & D     & D      &  D           & D      \\ 
\hline
Infinite   & D     & U      &  U           & U~\cite{Bulling/Farwer:10a}       \\ 
\hline
\end{tabular}
\caption{Decidability and undecidability results.}
\label{table:mc-results}
\end{figure}

\section{Conclusion} 
\label{sec:conclusion}
We have presented a model-checking algorithm for
\atlr, a logic with \emph{resource production}, which makes \atlr exceptional
in the landscape of resource logics, for most of which the model-checking
problem is undecidable \cite{Bulling/Farwer:10a,Bulling/Goranko:13a}.  
We compared \atlr with a similar logic (a variant of \ral, \cite{Bulling/Farwer:10a})
to understand the differences between the two logics and why the model-checking
problem for \atlr is decidable while the model-checking problem for \pfral with infinite semantics
is undecidable. As a by-product of this comparison, we show that the
model-checking problem for \pfral with finite semantics is decidable, solving
a problem left open in \cite{Bulling/Farwer:10a}.

Although the model-checking problem for \atlr in decidable, it is EXPSPACE-hard.
In future work, we plan to implement model-checking algorithms for feasible fragments 
of \atlr in the model-checker MCMAS \cite{Lomuscio//:09a}.

\paragraph{Acknowledgments}
This work was supported by the Engineering and Physical Sciences Research Council [grants EP/K033905/1 and EP/K033921/1]. We thank an anonymous reviewer for
a very thorough, detailed and constructive feedback.

\bibliographystyle{plain}
\bibliography{jlc14mc}
\end{document}